%% file: main.tex
\documentclass[sigconf]{acmart}
\AtBeginDocument{%
  }

\copyrightyear{2025}
\acmYear{2025}
\setcopyright{acmlicensed}\acmConference[CCS '25]{Proceedings of the 2025
ACM SIGSAC Conference on Computer and Communications Security}{October
13--17, 2025}{Taipei, Taiwan}
\acmBooktitle{Proceedings of the 2025 ACM SIGSAC Conference on Computer
and Communications Security (CCS '25), October 13--17, 2025, Taipei, Taiwan}
\acmDOI{10.1145/3719027.3765043}
\acmISBN{979-8-4007-1525-9/2025/10}

\usepackage[ruled]{algorithm2e}
\usepackage{float}
\usepackage{stmaryrd}
\input{macros}
\usepackage{booktabs}
\usepackage{subcaption}
\usepackage{listings}
\usepackage{graphicx}
\usepackage{framed}
\usepackage{multirow}
\usepackage{tikz}
\usetikzlibrary{arrows.meta}
\usepackage{subcaption}
\usepackage{diagbox} 

\usepackage{breqn}
\begin{document}

\title{Approximate Algorithms for Verifying Differential Privacy with Gaussian Distributions}

\author{Bishnu Bhusal}
\orcid{0000-0001-7522-5878}
\affiliation{%
  \institution{University of Missouri}
  \city{Columbia}
  \country{USA}
}
\email{bhusalb@missouri.edu}

\author{Rohit Chadha}
\orcid{0000-0002-1674-1650}
\affiliation{%
  \institution{University of Missouri}
  \city{Columbia}
  \country{USA}
}
\email{chadhar@missouri.edu}

\author{A. Prasad Sistla}
\orcid{0009-0005-8331-7912}

\affiliation{%
  \institution{University of Illinois at Chicago}
  \city{Chicago}
  \country{USA}
}
\email{sistla@uic.edu}

\author{Mahesh Viswanathan}
\orcid{0000-0001-7977-0080}
\affiliation{%
  \institution{University of Illinois Urbana-Champaign}
  \city{Urbana}
  \country{USA}
}
\email{vmahesh@illinois.edu}


\begin{abstract}
The verification of differential privacy algorithms that employ Gaussian distributions is little understood. 
This paper tackles the challenge of verifying 
such programs by introducing a novel approach to approximating probability distributions of loop-free programs that sample from both discrete and continuous distributions with computable probability density functions, including Gaussian and Laplace. We establish that verifying $(\epsilon,\delta)$-differential privacy for these programs is \emph{almost decidable}, meaning the problem is decidable for all values of $\delta$ except those in a finite set. 
Our verification algorithm is based on computing probabilities to any desired precision by combining integral approximations, and tail probability bounds. The proposed methods are implemented in the tool, {\ourtool}, using the {\FLINT} library for high-precision integral computations, and incorporate optimizations to enhance scalability. We validate {\ourtool} on fundamental privacy-preserving algorithms, such as Gaussian variants of the Sparse Vector Technique and Noisy Max, demonstrating its effectiveness in both confirming privacy guarantees and detecting violations.
\end{abstract}




\maketitle

\section{Introduction}
\label{sec:intro}
\input{sec_introduction}

\section{Preliminaries}
\label{sec:preliminaries}
\input{sec_prelim}

\section{Motivating Example: Sparse Vector Technique with Gaussians ({\SVTGauss})}
\label{sec:mot_example}
\input{sec_mot_example}

\section{Program syntax and semantics}
\label{sec:our_language}
\input{sec_our_lang}

\section[Checking differential privacy for DiPGauss  programs]{Checking differential privacy for $\ourlang$  programs}
\label{sec:checking_dp}
\input{sec_checking_dp}

\section{Approximating Output Probabilities}
\label{sec:compute_prob}

\input{sec_construct_graph}

\section{Implementation and Evaluation}
\label{sec:experiments}
\input{sec_experiments}

\section{Related Work}
\label{sec:related}
\input{sec_related}

\section{Conclusions and Future Work }
\label{sec:conclusion}
\input{sec_conclusions}

\paragraph*{Acknowledgements.} 
This work was partially supported by the National Science Foundation: Bishnu Bhusal and Rohit Chadha were partially supported by grant CCF 1900924, A. Prasad Sistla was partially supported by grant CCF 1901069, and Mahesh Viswanathan was partially supported by grants CCF 1901069 and CCF 2007428.


\bibliographystyle{ACM-Reference-Format}
\balance
\bibliography{main}
\ifdefined\short
\else
\clearpage
\appendix

\section{Pseudocode of Examples}
 \label{app:psuedocode}
 \input{app_examples}

 \section{Full Experimental Results}
 \label{app:fullexperimentalresults}
 \input{app_experiments}
\fi
\end{document}
\endinput


%% file: macros.tex
\usepackage{amsmath,amsfonts}
\usepackage{graphicx}
\usepackage{textcomp}
\usepackage{xcolor}
\usepackage[T1]{fontenc}
\usepackage[inline]{enumitem}
\usepackage{amsmath,amsfonts}
\usepackage{textcomp}
\usepackage{xcolor}
\usepackage{caption}
\usepackage{mathtools}
\usepackage{amsthm}

\theoremstyle{plain}
\newtheorem{theorem}{Theorem}
\newtheorem{lemma}[theorem]{Lemma}

\newtheorem{proposition}[theorem]{Proposition}
\theoremstyle{definition}
\newtheorem{definition}{Definition}
\newtheorem{example}{Example}

\theoremstyle{remark}
\newtheorem*{remark}{Remark}

\theoremstyle{remark}

\newcommand{\precision}{{\varrho}}
\newcommand{\minmax}{\textsc{Min-Max}}

\newcommand{\intexpr}{\mathcal{E}}

\newcommand{\exit}{\textsf{exit}}
\newcommand{\Range}{\textsf{Range}}

\newcommand{\SVT}{\textsf{SVT}}

\newcommand{\SVTGauss}{\textsf{SVT}-\textsf{Gauss}}
\newcommand{\AboveThreshold}{\textsf{SVT-Gauss-Ge}}
\newcommand{\BelowThreshold}{\textsf{SVT-Gauss}}

\newcommand{\LaplaceAboveThreshold}{\textsf{SVT-Laplace-Ge}}
\newcommand{\LaplaceBelowThreshold}{\textsf{SVT-Laplace}}

\newcommand{\MixOneAboveThreshold}{\textsf{SVT-Mix1-Ge}}

\newcommand{\MixTwoAboveThreshold}{\textsf{SVT-Mix2-Ge}}

\newcommand{\MixOneBelowThreshold}{\textsf{SVT-Mix1}}

\newcommand{\MixTwoBelowThreshold}{\textsf{SVT-Mix2}}

\newcommand{\dipc} {{\sf DiPC}}

\newcommand{\LaplaceNoisyMax}{\textsf{Noisy-Max-Laplace}}
\newcommand{\LaplaceNoisyMin}{\textsf{Noisy-Min-Laplace}}

\newcommand{\NoisyMax}{\textsf{Noisy-Max}}
\newcommand{\NoisyMin}{\textsf{Noisy-Min}}

\newcommand{\NoisyMaxGauss}{\textsf{Noisy-Max-Gauss}}
\newcommand{\NoisyMinGauss}{\textsf{Noisy-Min-Gauss}}

\newcommand{\timeout}{\mathsf{T.O}}
\newcommand{\na}{-}
\newcommand{\experimenterror} {\mathsf{O.M}}

\newcommand{\sumdeltamax}{\Delta_{max}}
\newcommand{\sumdeltamin}{\Delta_{min}}

\newcommand{\mylaplace}{\mathsf{Laplace}}
\newcommand{\gaussian}{\mathsf{Gaussian}}

\newcommand{\expr}{\mathsf{RExp}}

\newcommand{\lap}[1]{\textsf{Lap}(#1)}
\newcommand{\executions}[2]{\llbracket #1 \rrbracket_{#2}}

\newcommand{\mean}{\mu}
\newcommand{\std}{\sigma}
\newcommand{\scale}{b}

\newcommand{\tl}[1]{\mathsf{tl}(#1)}

\newcommand{\distr}{\mathsf{Distr}}
\newcommand{\pstate}{\mathsf{st}}
\newcommand{\randomvar}{X}

\newcommand{\inputset}{\mathcal{I}}
\newcommand{\outputset}{\mathcal{O}}

\newcommand{\randompartialfunc}{\beta}

\newcommand{\dompartialfunc}{\alpha}

\newcommand{\evalconstant}{\mathsf{eval_c}}

\newcommand{\evalrandom}{\mathsf{eval_r}}

\newcommand{\Grandom}{G_{\mathsf{rand}}}

\newcommand{\sourcenodes}{\mathsf{Src}}

\newcommand{\Gconst}{G_{\mathsf{const}}}

\newcommand{\graph}{\mathcal{G}}

\newcommand{\rand}{\mathsf{rand}}

\newcommand{\mathematica}{Wolfram Mathematica\textregistered}
\newcommand{\dependencyGraph}{\mathcal{G}}

\newcommand{\BExp}{\mathsf{BExp}}

\newcommand{\unresolved}{\mathsf{Unknown}}
\newcommand{\ourtool}{\textsf{DiPApprox}}

\newcommand{\grammarskip}{\textsf{skip}}

\newcommand{\dv}{\mathsf{\bf x}}
\newcommand{\Prog}{\mathsf{P}}
\newcommand{\ProgEps}[1][\epsilon]{\Prog_{#1}}
\newcommand{\intprob}{I}

\newcommand{\diffprivate}{\textsf{DP}}
\newcommand{\notdiffprivate}{\textsf{Not\_DP}}

\newcommand{\run} {\rho}

\newcommand{\FLINT} {\textsc{FLINT}}
\newcommand{\PLY}{\textsc{PLY}}
\newcommand{\igraph}{\textsc{igraph}}

\newcommand{\Ifs}{\textsf{if}}
\newcommand{\Thens}{\textsf{then}}
\newcommand{\eds}{\textsf{end}}
\newcommand{\Elses}{\textsf{else}}

\newcommand{\ifstatement}[3]{\Ifs \, #1\, \Thens\, #2 \, \Elses\, #3\, \eds}

\newcommand{\cR}{\mathcal{R}}

\DeclarePairedDelimiter{\abs}{\lvert}{\rvert}
\newcommand{\integers}{\mathbb{Z}}
\newcommand{\Rats}{\mathbb{Q}}
\newcommand{\Reals}{\mathbb{R}}

\newcommand{\Nats}{\mathbb{N}}
 
 \newcommand{\sdom}{\mathsf{DOM}}

\newcommand{\euler}{e}
\newcommand{\eulerv}[1]{\ensuremath{\euler^{#1}}}

\newcommand{\cU}{\mathcal{U}}
\newcommand{\cV}{\mathcal{V}}

\newcommand{\cX}{\mathcal{X}}

\newcommand{\Prob}{\mathsf{Prob}}

\newcommand{\st}  {\mathbin{|}}

\newcommand{\gauss}[1]{\mathcal{N}{(#1)}}

\newcommand{\set}[1]{\{#1\}}

\newcommand{\true}{\mathsf{true}}
\newcommand{\false}{\mathsf{false}}

\newcommand{\rmv}[1] {}

\newcommand{\upperbound}{\mathsf{UB}}
\newcommand{\lowerbound}{\mathsf{LB}}

\newcommand{\compute}[1][\precision]{\mathsf{ComputeProb}_{#1}}
\newcommand{\computescale}[1][\precision]{\mathsf{ComputeScaleProb}_{#1}}
\newcommand{\store}{\mathsf{store}}
\newcommand{\stores}{\mathsf{store}\_\mathsf{scale}}

\newcommand{\partialfunc}{\rightharpoonup}

\newcommand{\prbfn}[1] {\Prob[#1]}

\newcommand{\bv}{\mathsf{b}}

\newcommand{\s}[1] {\mathsf{#1}}

\newcommand{\rv}{\mathsf{\bf r}}

\newcommand{\deltaadjs}[2]{\delta_{#1,#2}}

\newcommand{\pr}[1]{\Prob[#1]}

\newcommand{\ourlang}{\textsf{DiPGauss}}
\newcommand{\newlanggauss}{\textsf{DiPGauss}}
\newcommand{\VerifyDP}{\ensuremath{\mathsf{VerifyDP}_{\precision}}}
\newcommand{\VerifyDPPair}{\ensuremath{\mathsf{VerifyPair}}}
\newcommand{\VerifyDPFunc}[1]{\VerifyDPPair(#1)}

\newcommand{\GenExpr}{\textsf{GenExpr}}
\newcommand{\GenExprFunc}[1]{\textnormal{\GenExpr(#1)}}

\makeatletter
\newcommand{\removelatexerror}{\let\@latex@error\@gobble}
\makeatother

\newcommand{\brown}[1]{\textcolor{brown}{#1}}

\newcommand{\unit}{\ensuremath{[0,1]}}
 
 \newcommand{\thr}{\mathsf{th}}
 \newcommand{\bpr}[1]{\mathsf{bpr}(#1)}
  \newcommand{\tpr}[1]{\mathsf{tpr}(#1)}
  \newcommand{\adjacent}{\Phi}

\newcommand{\Depsilon}{\ensuremath{\epsilon_{\mathsf{prv}}}}

\newcommand{\mRange}{$m$-\textsf{Range-Gauss}}  

\newcommand{\kminmax}{$k$-\textsf{Min-Max-Gauss}}

\newcommand{\NonPrivAboveThresholdOne}{\textsf{SVT-Gauss-Ge-Leaky-1}}

\newcommand{\NonPrivBelowThresholdOne}{\textsf{SVT-Gauss-Leaky-1}}

\newcommand{\NonPrivAboveThresholdTwo}{\textsf{SVT-Gauss-Ge-Leaky-2}}

\newcommand{\NonPrivBelowThresholdTwo}{\textsf{SVT-Gauss-Leaky-2}}

\newcommand{\NonPrivBelowThresholdLaplace}{\textsf{SVT-Laplace-Leaky}}

\newcommand{\criticaldef}[4]{\ensuremath{\mathsf{D}_{#1,#2,#3,#4}}}
\newcommand{\critical}{\criticaldef{\Prog}{\epsilon}{\Depsilon}{\adjacent}}

\newcommand{\CheckDP}{{\sf CheckDP}}

\newcommand{\LaplaceMrange}{$m$-\textsf{Range-Laplace}}

\newcommand{\LaplaceKminmax}{$k$-\textsf{Min-Max-Laplace}}

%% file: sec_introduction.tex
Differential Privacy~\cite{dmns06} 
is increasingly being adopted as a standard method for maintaining the privacy of individual data in sensitive datasets. In this framework~\cite{DR14}, data is managed by a trusted curator and queried by a potentially dishonest analyst. The goal is to ensure that query results remain nearly ``unchanged'' whether or not an individual's data is included, thereby preserving privacy even against adversaries with unlimited computational power and auxiliary data. However, it is challenging to design protocols that meet this high bar because reasoning about differential privacy is subtle, complex, and involves precise, quantitative reasoning. Many proposed algorithms and proofs have been found flawed~\cite{hr10,gru12,cm15,lyu2016understanding}, leading to growing interest in programming languages, type systems, and formal verification tools for privacy analysis~\cite{mc09,rp10,airavat,gupt,bkob12,dfuzz,Ba14,Ba15,Ba16,bgghs16,bfgaghs16,ZK17,LiuWZ18,AH18,WDWKZ19,GaboardiNP19,bcjsv20,bcksv21,csv21,sato-popl,sato-ahl,csv21,csvb23}.

Despite advances, automated verification of differential privacy still faces significant challenges. Differential privacy is often ensured by adding noise, commonly from Laplace and Gaussian (or Normal) distributions. While there exist several automated tools for verifying differential privacy with Laplace distributions~\cite{DingWWZK18,bcjsv20,CheckDP,BichselGDTV18,csv21,csvb23}, we are not aware of any automated tool that can effectively analyze programs that incorporate Gaussian noise. 


This paper studies the problem of automatically verifying if a program using Gaussian distributions meets the requirement of differential privacy. While the problem of checking differential privacy is undecidable even for programs that only toss fair coins~\cite{bcjsv20}, a decidable subclass of programs has been identified~\cite{bcjsv20}. However, programs in this subclass are not allowed to sample from Gaussian distributions; they have inputs and outputs over finite domains, and can sample from 
Laplace distributions. The algorithm in~\cite{bcjsv20} requires handling symbolic expressions with integral functionals, and does not generalize to Gaussian distributions as the Gaussian distribution lacks closed-form integral-free expressions for the cumulative distribution.  This paper presents an approach to reason about programs that can also sample from the Gaussian distributions. We do assume inputs and outputs are over finite domains, as in~\cite{bcjsv20}.  

Before presenting our contributions, let us recall the basic setup of differential privacy. The differential privacy program/mechansim is usually parameterized by $\epsilon,$ henceforth referred to as \emph{privacy parameter} which controls the noise added during the computation. There are two additional quantities that influence the protocol and its correctness: a \emph{privacy budget} $\Depsilon > 0$, and \emph{error parameter} $\delta \in \unit$, which accounts for the probability of privacy loss.\footnote{Often, $\Depsilon$ is expressed as a function of $\delta,\epsilon.$ A popular choice of $\Depsilon$ is $\epsilon$, which often occurs in case of pure differential privacy.} 
Given a binary relation $\adjacent$ on inputs called an \emph{adjacency relation}, and taking $\prbfn{\epsilon,\Prog(u)\in F}$ to be the probability that $\Prog$ outputs $o \in F$ on input $u$, $\Prog$ is said to be \emph{$(\Depsilon,\delta)$-differentially private} if for each pair $(u,u')\in \adjacent$, and measurable subset of outputs $F$, we have that 
\begin{equation} \label{eq:diffprivacy}\prbfn{\epsilon, \Prog(u)\in F}-\eulerv{\Depsilon}\prbfn{\epsilon,\Prog(u')\in F} \leq \delta.\end{equation}
When $\delta=0$, $(\Depsilon,0)$-differential privacy is referred to as pure $\Depsilon$-differential privacy.  
We shall say that $\Prog$ is \emph{$(\Depsilon,\delta)$-strictly differentially private} if the inequality in (\ref{eq:diffprivacy}) holds \emph{strictly} (i.e, we require $<$ instead of $\leq$).

\rmv{
Assume the program $\Prog$ takes inputs and outputs from finite domains. Thus, if there was an {\lq\lq}oracle{\rq\rq} that outputs $\prbfn{\epsilon,\Prog,u,o}$, the probability of the program $\Prog$ outputting $o$ on input $u,$ we can check $(\Depsilon,\delta)$-differential privacy: enumerate all adjacent input pairs $u,u',$ enumerate all possible subsets of outputs $S$, and check that $\sum_{u\in S}\prbfn{\epsilon,\Prog,u,S}\leq \eulerv \epsilon \prbfn{\epsilon,\Prog,u',S} +\delta$. If the continuous distributions used in $\Prog$ are all Laplacian distributions, this {\lq\lq}oracle{\rq\rq} does indeed exist~\cite{bcjsv20}, and are ratios of polynomials in $\eulerv{\epsilon}$. Now, the resulting formulas can be checked by appealing to results in~\cite{mccallum2012deciding}. However, such an oracle may not exist when the values sampled are from the Gaussian (normal) distribution. For example, no known integral-free expression exists for the cumulative distribution of the Gaussian distribution. Thus, even the probability of a normally distributed random variable $\cX$ being bigger than a given constant $c$ cannot be written in an integral-free form. 
}
\sloppy
\paragraph*{Contributions}
Given fixed $\epsilon>0$, let $\prbfn{\epsilon,u,o,\Prog}=\prbfn{\epsilon,\Prog(u)\in \set{o}}$ denote the probability of $\Prog$ returning a single output $o$ on input $u.$
Let $\compute(\epsilon,u,o,\Prog)$ be a function that returns an interval $[L,U]$ such that $\prbfn{\epsilon,u,o,\Prog} \in [L,U]$ and $|U-L| \leq {2^{-\precision}}$; in other words, $\compute(\cdot)$ approximates $\prbfn{\epsilon,u,o,\Prog}$ with desired precision $\precision$. Our first result establishes that if the function $\compute(\cdot)$ is computable for a class of programs, then the problem of determining if program $\Prog$ in this class is non-$(\Depsilon,\delta)$-differential privacy is recursively enumerable (See Theorem~\ref{thm:decidability} on Page \pageref{thm:decidability}). Furthermore, we show that the computability of $\compute(\cdot)$ also implies that checking if a $\Prog$ is $(\Depsilon,\delta)$-strict differentially private, is recursively enumerable (See Theorem~\ref{thm:decidability} on Page \pageref{thm:decidability}). These two observations suggest that verifying $(\Depsilon,\delta)$-differential privacy is \emph{almost decidable} --- the problem of checking if $\Prog$ is $(\Depsilon,\delta)$-differentially private is \emph{decidable} for all values of $\delta$, \emph{except} those in a finite set $\critical$ that depends on the program $\Prog,$ $\epsilon,\Depsilon$ and $\adjacent$. Informally, $\critical$ is the set of $\delta$ for which the inequality (\ref{eq:diffprivacy}) is an equality; the precise definition is given in Definition~\ref{def:critical} on Page~\pageref{def:critical}.

Next, we observe that the function $\compute(\epsilon,u,o,\Prog)$ is computable for a large class of programs that sample from distributions such as the Gaussian distribution. Our class, called {\newlanggauss}, consist of \emph{loop-free} programs, with finite domain inputs and outputs. These programs can sample from 
continuous distributions, provided the distributions have \emph{finite} means and \emph{finite} variances~\footnote{A continuous distribution may not have finite mean or finite variance.} and \emph{computable}~\footnote{By computable, we mean computable as defined in recursive real analysis~\cite{ko-real}.} probability density functions. This includes commonly used distributions like Gaussian and Laplace. This class of programs differs from the class presented in~\cite{bcjsv20} --- while they are loop-free, they allow sampling from Gaussian distributions. Despite their simple structure, {\newlanggauss} programs include many widely used algorithms such as the Sparse Vector Technique (SVT)~\cite{svt}, Noisy Max~\cite{DingWWZK18} and their variants with Gaussian mechanisms~\cite{svt-gauss}. These algorithms are used in applications like ensuring the privacy of Large Language Models (LLMs)~\cite{google-svt-llm}.

To describe our algorithm for $\compute(\epsilon,u,o,\Prog)$, observe that  $\prbfn{\epsilon,u,o,\Prog}$ can be written as a sum of the probabilities of \emph{execution paths} of $\Prog$ on input $u$ leading to output $o$. Since $\Prog$ is loop-free, $\Prog$ has only a finite number of execution paths on each input $u$. The probability of each execution path can be written as sum of iterated integrals. If none of the integrals have $\infty$ or $-\infty$ as upper or lower limits, given the assumption that the probability density functions are computable, these integrals can be evaluated to any desired precision. 

We can avoid having $\infty$ or $-\infty$ as limits of integrals by appealing to tail bounds derived from Chernoff or Chebyshev inequalities as follows. 
Given a threshold $\thr>0$, we can write the probability of an execution path $\tau$ as 
$\bpr{\epsilon,\tau,\thr} +  \tpr{\epsilon,\tau,\thr}$. Here $\bpr{\cdot}$ is the probability of the execution $\tau$ under the constraint that all sampled values $\randomvar_{\rv_1},\ldots,\randomvar_{\rv_n}$ in the execution remain within $\thr\cdot\std_i$ from $\mean_i$ where $\mean_i$ and $\std_i$ are the  mean and standard deviation of the distribution of $\randomvar_{\rv_i}$, respectively. The term $\tpr{\cdot}$ accounts for the tail probability, capturing cases where at least one sampled value deviates by at least $(\thr\cdot\std)$ from $\mean$. Now, $\bpr{\epsilon,\tau,\thr}$ can again be written as a sum of iterated integrals with finite limits. The assumption of computability ensures that $\bpr{\cdot}$ can be computed to any desired precision.  Moreover, by Chebyschev's inequality, we can always select $\thr$ such that $\tpr {\epsilon,u,o,\thr}$ is arbitrarily small. This guarantees that $\prbfn{\epsilon,u,o,\Prog}$ can be computed to any required degree of precision.

Using our algorithm for $\compute(\epsilon,u,o,\Prog)$ and algorithm for almost deciding $(\Depsilon,\delta)$-differential privacy, we implement an algorithm $\VerifyDP$ that returns one of 3 outputs: $\diffprivate$, $\notdiffprivate$, or $\unresolved$. Crucially, as the precision parameter $\precision$ decreases, the likelihood of getting $\unresolved$ diminishes. Our implementation of {\VerifyDP} incorporates several optimizations to improve scalability for practical examples. 

 First, differential privacy requires verifying that inequality (\ref{eq:diffprivacy}) holds for \emph{all} pairs of adjacent inputs $u,u'$ and \emph{all} subsets $F$ of outputs. A key insight is that this verification can be reduced to checking a modified equation (See Lemma~\ref{lem:DiffPrivacyRephrased} on Page~\pageref{lem:DiffPrivacyRephrased}) where $F$ is taken as the set of all outputs, leading to \emph{exponential} 
reduction in complexity. Intuitively, it is enough to check the inequality (\ref{eq:diffprivacy}) for the set $F$ of all outputs $o$ for which  $\prbfn{\epsilon, \Prog(u)=o}-\eulerv{\Depsilon}\prbfn{\epsilon,\Prog(u')=o}> 0.$

Next, we observe that the integral nesting depth significantly impacts performance while computing the value of $\compute(\epsilon,u,o,\Prog).$
Thus, we introduce a heuristic to reduce it. Intuitively,  each variable being integrated in an integral $I$ during  $\compute(\epsilon,u,o,\Prog)$ represents a value sampled from an independent distribution in the program. We call such variables, independent random variables. For each nested integral $I$, we can define a directed acyclic graph, called the \emph{dependency graph} of $I$, on the independent random variables. An edge from $\rv$ to $\rv'$ in the dependency graph indicates that a) the integration over variable $\rv'$ is nested within integration over $\rv$ in $I,$ and that b) $\rv$ appears in either the upper or lower limit in the integration over $\rv'.$ 

We observe that if the sets of variables reachable from two variables $\rv_1$ and $\rv_2$ in the dependency graph of $I$ are disjoint, the nested depth of $I$ can be reduced by separating the integrations over $\rv_1$ and $\rv_2$. Exploiting this observation, we give an algorithm based on the topological sorting of the dependency graph which yields an integral expression $I'$ that has the same value as $I,$ but lower nesting depth.  
Applying this heuristic, our prototype tool was able to rewrite the integrals for the {\SVTGauss} algorithm~\cite{svt-gauss} with at most $3$ nested integrals, independent of the number of input variables. Without this optimization, the nested depth scales with the number of input variables. 

We implemented $\VerifyDP$ in a tool, {\ourtool}. Our implementation leverages the {\FLINT} library~\cite{flint} to evaluate nested integrals with finite limits. {\FLINT} enables rigorous arithmetic with arbitrary precision using ball enclosures, where values are represented by a midpoint and a radius. Using the above outlined approach, our tool successfully verified variants of the Sparse Vector Technique~\cite{svt,svt-gauss}, Noisy Max~\cite{DR14}, $k$-\minmax~\cite{csvb23}, $m$-\Range~\cite{csvb23} with Gaussians (and Laplacians) across varying input lengths. Additionally, it confirmed the non-differential privacy of insecure versions of the Sparse Vector Technique with Gaussians. These experiments highlight the potential of our approach for the automated verification of differential privacy in programs that sample from complex continuous distributions.  {\ourtool} is available for download at ~\cite{artifact}.

\paragraph{Organization} The rest of the paper is organized as follows. Preliminary mathematical notation and definitions are given in Section~\ref{sec:preliminaries}. Section~\ref{sec:mot_example} discusses the variant of Sparse Vector Technique with Gaussians~\cite{svt-gauss}. Section~\ref{sec:our_language} introduces the syntax and semantics of {\ourlang}, the language we use to specify differential privacy algorithms. Section~\ref{sec:checking_dp} presents $\VerifyDP$, assuming that there is an algorithm that computes $\compute(\cdot).$ It also presents the rephrasing of differential privacy. Section~\ref{sec:compute_prob} shows how $\compute(\cdot)$ can be implemented and presents the integral optimization discussed above. We also present our decidability results here in this section. Section~\ref{sec:experiments} discusses {\ourtool} and presents the experimental evaluation. Section~\ref{sec:related} discusses related work, and our conclusions are presented in Section~\ref{sec:conclusion}.  

An extended abstract of the paper~\cite{CCSVersion} is going to appear in the  Proceedings of the 2025 ACM SIGSAC Conference on Computer and Communications Security (CCS' 25).
This paper consists of details of the experiments omitted in~\cite{CCSVersion}.

\rmv{
\prasad{An alternate way of stating it is as follows: There exists a co-finite set $W_{\Prog,\epsilon}\subseteq \unit$ such that for all $\delta$ in this set, checking $(\Depsilon,\delta)$ differential privacy is decidable. Stating like this may be better as this corresponds to decidability for almost all values of $\delta\in \unit.$}
\prasad{We dont have to specifyy the set $U_{\Prog,\epsilon}$. We can point to a place on the paper where it is specified. Or we can say that for a given $\Prog$ and $\epsilon$, $U_{\Prog,\epsilon}$ is the set of all $\delta$ for which the inequality (1) becomes an equality.}
}


%% file: sec_prelim.tex

We denote the sets of real, rational, natural, and integer numbers by $\Reals$, $\Rats$, $\Nats$, and $\integers$, respectively, and the Euler constant by $\euler$. A \emph{partial function} $f$ from $A$ to $B$ is denoted as $f: A \partialfunc B$. We assume that the reader is familiar with probability. For events $E$ and $F$, $\pr{E}$ represents the probability of $E$, and $\pr{E | F}$ the conditional probability of $E$ given $F$. 

\paragraph*{Gaussian (Normal) Distribution} The one dimensional Gaussian distribution, denoted by $\gauss{\mu, \sigma}$, is parameterized by the mean $\mean \in \Reals$ and the standard deviation $\std>0$.  Its probability density function (PDF), $f_{\mu, \sigma}(x)$, is defined as:
\begin{equation}
f_{\mean, \std}(x) = \frac{1}{\std \sqrt{2\pi}} \eulerv{-\frac{(x - \mean)^2}{2\std^2}}
\end{equation}

\paragraph*{Laplace distribution}
The Laplace distribution denoted $\lap{\mean, \scale}$ is parameterized by mean $\mean$ and the scaling parameter $b \geq 0$. Its probability density function (PDF), $g_{\mean, \scale}(x)$, is defined as:
\begin{equation}
 g_{\mean, b}(x) = \frac{1}{2\scale}\eulerv{-\frac{|x-\mean|}{\scale}}
\end{equation}
The standard deviation of $\lap{\mean,\scale}$ is ${\sqrt 2 b}.$

\paragraph*{Tail Bounds} 
In our framework, execution probabilities are computed using nested definite integrals. Apriori, integrals may involve $\infty$ and $-\infty$ as upper and lower limits, respectively. As we will see (Section~\ref{sec:compute_prob}), we shall use tail bounds on probabilities to approximate the integral as a sum of definite integrals with finite limits. 

Given a random variable $\randomvar$  with finite mean $\mean$ and non-zero standard deviation $\std$ and $\thr>0,$ we say that the \emph{two-sided tail probability} is 
$$\tl{\thr,\randomvar,\mean,\std}=\pr{\abs{\randomvar-\mean}>\thr\cdot\std}.$$ 
For Laplace distribution, we have that 
$ \tl{\thr,\randomvar,\mean,\std}= \eulerv{- {\thr{\sqrt 2}}}.$

If $\randomvar$ is a Gaussian (also known as normal) random variable, we obtain the following bounds on the tail probabilities using Chernoff bounds:
$\tl{\thr,\randomvar,\mean,\std} \leq 2\eulerv{-\frac{\thr^2}{2}}.$

For a general random variable with finite mean $\mean$ and standard deviation $\std$, Chebyschev's inequality yields that  $\tl{\thr,\randomvar,\mean,\std}\leq \frac 1 {\thr^2}.$
Observe that all the bounds on tail probabilities discussed here monotonically decrease to $0$ as  $\thr$ tends to $\infty.$

\paragraph*{Approximate Differential Privacy} 

Differential privacy~\cite{dmns06} is a framework that enables statistical analysis of databases containing sensitive personal information while protecting individual privacy. A randomized algorithm $\Prog$, called a \emph{differential privacy mechanism}, mediates the interaction between a (potentially dishonest) data analyst and a database $D$. The analyst submits queries requesting aggregate information such as means, and for each query, $\Prog$ computes its response using both the actual database values and random sampling to produce ``noisy'' answers. While this approach ensures privacy, it comes at the cost of reduced accuracy. The amount of noise added by $\Prog$ is controlled by a \emph{privacy parameter} $\epsilon > 0$.

The framework provides privacy guarantees for all individuals whose information is stored in database $D$. This is informally captured as follows. Let $D \setminus \set{i}$ denote the database with individual $i$'s information removed. A secure mechanism $M$ ensures that for any individual $i$ in $D$ and any sequence of possible outputs $\overline{o}$, the probability of $\Prog$ producing $\overline{o}$ remains approximately the same whether querying $D$ or $D \setminus \set{i}$.

Let us define this formally. A differential privacy mechanism is a family of programs $\ProgEps$ whose behavior depends on the privacy parameter $\epsilon$; for notational simplicity, we will often drop the subscript and use $\Prog$ to refer to the programs. Given a set $\cU$ of inputs and a set $\cV$ of outputs, a randomized function $\Prog$ from $\cU$ to $\cV$ takes an input in $\cU$ and returns a distribution over $\cV$. For a measurable set $F\subseteq \cV$, the probability that the output of $\Prog$ on $u$ is in $F$ is denoted by $\pr{\Prog(u)\in F}$ . We assume that $\cU$ is equipped with a binary asymmetric relation $\adjacent \subseteq \cU\times \cU$, called the \emph{adjacency relation}. \emph{Adjacent} inputs $(u_1,u_2) \in \adjacent$ represent query outputs for databases $D$ and $D \setminus \set{i}$ where $i$ is some individual.
\begin{definition} \label{def:differential-privacy} Let $\epsilon,\Depsilon > 0$, $0\leq\delta\leq 1$ and $\adjacent\subseteq \cU\times \cU$ be an adjacency relation. Let $\Prog$ be a randomized algorithm depending on a privacy parameter $\epsilon$ with inputs $\cU$ and outputs $\cV$. We say that $\Prog$ is $(\Depsilon,\delta)$-differentially private with respect to $\adjacent$ if for all measurable subsets $S\subseteq \cV$ and $u,u' \in \cU$ such that $(u,u')\in \adjacent$,
\[
\pr{\Prog(u)\in S} \leq \eulerv{\Depsilon}   \, \pr{\Prog(u')\in S} +\delta.
\]
$\Depsilon$ is called the \emph{privacy budget}, and $\delta$ the \emph{error parameter}. 

We say that $\Prog$ is $(\Depsilon,\delta)$-strictly differentially private with respect to $\adjacent$ if for all measurable subsets $S\subseteq \cV$ and $u,u' \in \cU$ such that $(u,u')\in \adjacent$,
\[
\pr{\Prog(u)\in S} < \eulerv{\Depsilon}   \, \pr{\Prog(u')\in S} +\delta.
\]
\end{definition}


%% file: sec_mot_example.tex
Let us walk through a simple example to demonstrate our method before exploring the full mathematical details.
The Sparse Vector Technique (SVT) is a fundamental algorithmic tool in differential privacy that plays an important role in adaptive data analysis and model-agnostic private learning~\cite{DNRRV09, DR14}. While the original Sparse Vector Technique (SVT) uses Laplace noise to achieve $(\Depsilon,0)$-differential privacy, recent work has explored a Gaussian variant that achieves $(\Depsilon,\delta)$-differential privacy~\cite{svt-gauss}. SVT with Gaussian distribution ({\SVTGauss}) offers better utility than the version using Laplace noise through more concentrated noise~\cite{svt-gauss}.
%

\begin{algorithm}
\DontPrintSemicolon
\SetAlgoLined

\KwIn{$q[1:N]$}
\KwOut{$out[1:N]$}
\;
$\rv_T \gets \gauss{T,  \frac {2}{\epsilon}}$\;
  \For{$i\gets 1$ \KwTo $N$}
  {
    $\rv \gets \gauss{q[i], \frac{4}{\epsilon}}$\;

    \uIf{$ \rv \geq \rv_T $}{
      $out[i] \gets 1$\;
      exit\;
      }
    \Else{
      $out[i] \gets 0$}
  }

\caption{\SVTGauss}
\label{fig:SVT}
\end{algorithm}

The SVT mechanism with Gaussian sampling ({\SVTGauss}) is shown in Algorithm~\ref{fig:SVT}. Given an array $q$ containing answers to $N$ queries and threshold $T$, the goal is to output the index of the first query that exceeds the threshold in a privacy preserving manner~\footnote{In general, the SVT protocols identifies the first $c$ queries that exceed the threshold, for some fixed parameter $c$. Also, the privacy of the algorithm can only be guaranteed for query answers that are \emph{sensitive} upto some parameter $\Delta$. Both $c$ and $\Delta$ influence the noise that is added when processing the array $q$. In Algorithm~\ref{fig:SVT}, we assume that $c = 1$ and $\Delta = 1$.}. The algorithm perturbs the threshold and each query answer by adding Gaussian noise. The algorithm progressively compares the perturbed query answer against the perturbed threshold, assigning $0$ to the output array $out$ if the query is less and $1$ if it is not. The algorithm stops when either $1$ is assigned or if all the query answers in $q$ are processed. $\epsilon$ is the privacy parameter of the algorithm.

The input set $\cU$ consists of $N$-length vectors $q$, where the $k$th element $q[k]$ represents the answer to the $k$th query on the original database. The adjacency relation $\Phi$ on
inputs is defined as follows: $q_1$ and $q_2$ are adjacent if and only if $\abs{q_1[i]-q_2[i]}\leq 1$ for each $1\leq i \leq N$. 

\rmv{
The algorithm takes as input an array $q$ of length $N$, where $N$ is the total number of queries and each element holds a query answer. The output array $out$ is initially filled with $\bot$ (False) values, where $\top$ represents True. In the Sparse Vector Technique (SVT), $\top$ answers consume most of the privacy budget, allowing only $c$ such answers before the budget is exhausted~\cite{DNRRV09,ZK17}, while there is no limit on $\bot$ answers. The SVT algorithm is parameterized by the privacy budget $\epsilon$, defining a class of programs for each $\epsilon>0$.

}

Recall that the pdf of a Gaussian random variable $\gauss{\mu,\sigma}$ is denoted as $f_{\mu,\sigma}.$
Consider the case when $T=0$, $N=2$ and entries in $q$ are limited to $\set{0,1}$. Thus, there are $4$ possible inputs ($[0,0]$, $[0,1]$, $[1,1]$, and $[1,0]$) and three possible outputs ($[0,0]$, $[1,0]$, and $[0,1]$). Let $X_T$, $X_1$, $X_2$ be the random variables denoting the values of variables $\rv_T$, $\rv$ during different iterations, in Algorithm~\ref{fig:SVT}. Observe that $X_T$ is drawn from $\gauss{0,\frac{2}{\epsilon}}$ and $X_i$ from $\gauss{q[i],\frac{4}{\epsilon}}$. 

Consider adjacent inputs $[0,1]$ and $[1,1].$
On input $q = [0,1]$, the probability of output $[0,0]$ is given by $\pr{X_1<X_T, X_2 < X_T}$, which can be computed as:
\[
 p_1(\epsilon) =\!\!\int_{-\infty}^{\infty}\!\! f_{0, \frac{2}{ \epsilon}}(x_T) \! \int_{-\infty}^{x_T}\!\! f_{0,\frac{4}{ \epsilon}}(x_1) \! \int_{-\infty}^{x_T}\!\! f_{1, \frac{4}{\epsilon}}(x_2)\ dx_2 \!\ dx_1 \!\ dx_T
\]
Similarly, the probability of output $[0,1]$ on input $[0,1]$ is given by $\pr{X_1<X_T, X_2 > X_T}$, which can be computed as:
\[
 p_2(\epsilon) =\!\!\int_{-\infty}^{\infty}\!\! f_{0, \frac{2}{ \epsilon}}(x_T) \! \int_{-\infty}^{x_T}\!\! f_{0,\frac{4}{ \epsilon}}(x_1) \! \int_{x_T}^{\infty}\!\! f_{1, \frac{4}{\epsilon}}(x_2)\ dx_2 \!\ dx_1 \!\ dx_T
\]

In the same way, when the input is $[1,1]$, the probability of output $[0,0]$ and $[0,1]$ is given by:
\begin{align*}
p_1'(\epsilon)=\!\!\int_{-\infty}^{\infty}\!\! f_{0, \frac{2}{\epsilon}}(x_T) \!\int_{-\infty}^{x_T}\!\! f_{1,\frac{4}{ \epsilon}}(x_1)\!\int_{-\infty}^{x_T}\!\! f_{1,\frac{4}{\epsilon}}(x_2)\ dx_2 \!\ dx_1 \!\ dx_T\\
p_2'(\epsilon)=\!\!\int_{-\infty}^{\infty}\!\! f_{0, \frac{2}{\epsilon}}(x_T) \!\int_{-\infty}^{x_T}\!\! f_{1,\frac{4}{ \epsilon}}(x_1)\!\int_{x_T}^{\infty}\!\! f_{1,\frac{4}{\epsilon}}(x_2)\ dx_2 \!\ dx_1 \!\ dx_T
\end{align*}

Observe that $p_1(\epsilon)$, $p_2(\epsilon)$, $p_1'(\epsilon)$ and $p_2'(\epsilon)$ are functions of $\epsilon$. To check if the adjacent inputs $[0,1]$ and $[1,1]$ satisfy the conditions of $(\Depsilon,\delta)$-differential privacy for given privacy budget $\Depsilon$, error $\delta$ and output set $\set{[0,0], [0,1]}$, we need the following  to hold.
\begin{align*}
p_1(\epsilon) + p_2(\epsilon) &\leq e^{\Depsilon} [p_1'(\epsilon) + p_2'(\epsilon)] + \delta, \\
p_1'(\epsilon) + p_2'(\epsilon) &\leq e^{\Depsilon} [p_1(\epsilon) + p_2(\epsilon)] + \delta.
\end{align*}
Note that since outputting $[0,0]$ and $[0,1]$ are independent events, we can sum their probabilities to obtain an expression for $\set{[0,0], [0,1]}$. 

Verifying $(\Depsilon,\delta)$-differential privacy for a given $\epsilon>0$ involves computing expressions like $p_i(\epsilon)$ and $p_i'(\epsilon)$ and checking if inequalities like the one above hold for all possible sets of outputs.

The following theorem states the differential privacy of Algorithm~\ref{fig:SVT}, and follows from the results of~\cite{svt-gauss}.
\begin{theorem}
\label{thm:gauss}
For any $\epsilon > 0$ and $0 < \delta \leq \frac 1 {1+N},$ {\SVTGauss} (Algorithm~\ref{fig:SVT}) is $(\Depsilon, \delta)$-differential privacy for any $\Depsilon$ such that
\[
\Depsilon \geq \frac{5\epsilon^2}{32} + \frac{\sqrt{5}}{2}\epsilon\sqrt{\log{\frac{1}{\delta}}}.
\]
\end{theorem}
{\SVTGauss} belongs to the class of programs that we consider in this paper. Observe that when $\epsilon = 0.5,\delta = 0.01, N<100$, $\Depsilon \geq 1.24$. In our experiments, we are able to automatically verify differential privacy with these values of $\epsilon,\Depsilon$ and $N\leq 5.$ When we consider only single pair of adjacent inputs, our tool is able to handle $N$ upto $25.$ 

%% file: sec_our_lang.tex
We introduce a class of probabilistic programs called {\ourlang}, where variables can be assigned values drawn from Gaussian distributions or Laplace distributions, commonly used in differential privacy algorithms. {\ourlang} is designed with syntactic restrictions that simplify its encoding into integral expressions. While these restrictions impose certain limitations, they also enable definitive verification of whether a program satisfies differential privacy. Despite its constraints, {\ourlang} is a powerful language capable of expressing interesting differentially private algorithms.

Before we present our syntax formally, we observe that differential privacy algorithms are often described as parameterized programs. Colloquially, this parameter is the privacy budget. However, in many instances, the parameter may be different from the privacy budget. (See Theorem~\ref{thm:gauss}, Section~\ref{sec:mot_example}.) Thus, to explicitly distinguish the parameter in the program and the privacy budget, we shall refer to the program parameter as privacy parameter and denote it as $\epsilon.$ We will use $\Depsilon$ to denote the privacy budget.

\begin{figure}[htp]
\begin{framed}
\raggedright

Expressions ($\rv\in \cR, \dv\in \cX, q \in \Rats, \sim\in \set{=,\leq,<,\geq,>,\neq}$):
$$\begin{array}{lll}
R & := & \rv \mid qR \mid R + R  \mid R + q \\
B & := & R \sim R \mid R \sim  \dv \mid \dv \sim \dv 
\end{array}$$

Program Statements ($d\in \sdom, a\in \Rats^{>0}$):
$$\begin{array}{lll}
S & := & \dv \gets d \\
     & \mid & \rv \gets \gauss{\dv, \frac{a}{\epsilon}} \\ & \mid & \rv \gets \lap{\dv, \frac{a}{\epsilon}} \\
     & \mid & \rv \gets R \\
     & \mid & \ifstatement B S S \\
     & \mid & \grammarskip \\ & \mid & S; S
\end{array}$$
\end{framed}
\vspace{-10pt}
\captionsetup{font=footnotesize,labelfont=footnotesize}
 \caption{\footnotesize BNF grammar for $\ourlang$. $\sdom$ is a finite discrete domain, taken to be a finite subset of the rationals.  $\cR$ is the set of real random variables and $\cX$ is the set of $\sdom$ variables. $\Rats^{>0}$ denotes set of positive rational numbers.
 }
\label{fig:BNFFull}
\end{figure}


\subsection{Syntax of {\ourlang} Programs}
\label{sec:syntax}

The formal syntax of {\ourlang} programs is presented in Figure~\ref{fig:BNFFull}. {\ourlang} programs are parametrized, loop-free programs that can sample from continuous distributions. Programs have two types of variables: real random variables and finite-domain variables from $\sdom$, denoted by sets $\cR$ and $\cX$ respectively. We assume that $\sdom$ is some finite subset of rationals and so they can be compared against each other and with real values. Boolean expressions ($B$) can be constructed by comparing real variables with each other, with $\sdom$-variables, or by comparing $\sdom$ variables with each other. 

A program is a sequence of statements that can either be assignments to program variables or if conditionals. Assignments can either assign constants (real or $\sdom$ values) or values drawn from continuous distributions. In Figure~\ref{fig:BNFFull}, the only distributions we have listed are the Laplace or Gaussian distributions. This is done to keep the presentation simple in this paper. Our results apply even when the syntax of the program language is extended, where samples are drawn from any continuous distribution with a finite mean and variance and a computable probability density function.  

\begin{remark}
    A couple of remarks on the program syntax are in order at this point. {\ourlang} does not natively support loops but for-loops can be seen as syntactic sugar in the standard way. A loop of the form \textit{for $i = 1$ to $N$ do $S$} can be expanded into a sequence $S_1; S_2; \dots; S_N$, where $N$ is a constant and each iteration is explicitly unrolled.

    Next, Boolean expressions used in conditionals are restricted to comparison between program variables and constants; the syntax does not allow the use of standard logical operators such as negation, conjunction, and disjunction. However, this is not a restriction in expressive power. Taking a step based on the negation of a condition holding can be handled through the else branch, conjunction through nested if-thens, and disjunction through a combination of nesting and else branches.

    {Finally, due to space limitations, the paper focuses on the most relevant language features to illustrate our techniques. Our approach can be extended to support real variables and Gaussian sampling within the exact language framework of~\cite{bcjsv20}, resulting in decidability guarantees for this richer language. However, the methods in~\cite{bcjsv20} are not applicable to our extended setting, as they crucially depend on the existence of closed-form solutions for integrals, which are unavailable in our case.}

\end{remark}

A program $\Prog$ in {\ourlang} is defined as a triple $(\inputset, \outputset, S)$, where:
\begin{itemize}
    \item $\inputset \subseteq \cX$ is a set of private input variables.
    \item $\outputset \subseteq \cX$ is a set of public output variables.
    \item $\inputset \cap \outputset = \emptyset$
    \item $S$ is a program statement generated by the non-terminal $S$ of the grammar in Figure~\ref{fig:BNFFull}.
\end{itemize}
As seen in the grammar of Figure~\ref{fig:BNFFull}, each sampled probability distribution used in the statements of $\Prog$, has a parameter $\epsilon$. Thus, $\Prog$ is a parameterized program with parameter $\epsilon$ appearing in it. The parameter $\epsilon$ will be instantiated when computing the probabilities associated with $\Prog$.
\begin{remark}
Strictly speaking, $\Prog$ represents a family of programs, and it is more accurate to represent it as $\ProgEps.$ However, we choose to not mention $\epsilon$ explicitly to reduce notational overhead. 
\end{remark}

\begin{example}
\label{ex:svtlang}
Algorithm~\ref{fig:SVT} can be rewritten as a {\ourlang} program when $T = 0$ and $N = 2$. This is shown as Algorithm~\ref{fig:SVTG}, where the bounded for-loop is unrolled and written without any loops. 
\end{example}

\begin{algorithm}[h]
  \SetAlgoLined
  \DontPrintSemicolon
  \KwIn{$q_1,q_2$}
  \KwOut{$out_1,out_2$}
  \nl $T \gets 0$\;
  \nl $out_1 \gets 0$\;
  \nl $out_2 \gets 0$\;
  \nl $\rv_T \gets \gauss{T, \frac{2}{\epsilon}}$\;   
  \nl $\rv_1 \gets \gauss{q_1, \frac{4}{\epsilon}}$\; 
  \nl \uIf{$\rv_1 \geq \rv_T$}{
    \nl $out_1 \gets 1$\;
  }\Else{
    \nl $\rv_2 \gets \gauss{q_2, \frac{4}{\epsilon}}$\;
    \nl \If{$\rv_2 \geq \rv_T$}{
      \nl $out_2 \gets 1$\;
    }
  }
  \caption{
   {\SVTGauss} with $N=2$ written in {\ourlang}
  }
  \label{fig:SVTG}
\end{algorithm}

We conclude this section with a couple of assumptions about programs in {\ourlang}. We will assume that in every program, each real variable is assigned a value at most once along every control path. Clearly, since our program are loop-free, this is not a restriction, as a program where variables are assigned multiple times can be transformed into one that satisfies this assumption by introducing new variables. However, making this assumption about our programs will make it easier to describe the semantics.

Finally, we will assume that programs are \emph{well-formed}. That is, all references in non-input variables in the program are preceded by assignments to those variables. 


\subsection{Semantics}
\label{sec:semantics}



The semantics for {\ourlang} programs presented in this section crucially relies on the notion of state. Typically, state for a non-recursive imperative program without dynamic allocation is just an assignment of values to the program variables. However, for programs where real variables are assigned values from continuous distributions, this does not work. Instead, we will record the values of real variables ``symbolically'' --- we will either record the distribution (plus parameters like mean) from which the value of a real variable is sampled or the expression it is assigned~\footnote{Recall that we assume that every real variable is assigned at most once during an execution; see Section~\ref{sec:syntax}.}. However, to reliably assign probabilities to paths with such ``symbolic states'', we also need to track the Boolean conditions that are assumed to hold so far. Based on these intuitions let us define states formally.

\paragraph{States.}
Let $\BExp$ and $\expr$ denote the set of expressions derived from the non-terminals $B$ and $R$, respectively, in Figure~\ref{fig:BNFFull}. Define $\distr = \{\gaussian, \mylaplace \} \times \Rats \times \Rats^{>0}$; elements of this set denote distributions along with appropriate parameters. So, $(\gaussian, \mean, \std)$ represents the Gaussian distribution with mean \(\mean\) and standard deviation \(\std\) while \((\mylaplace, \mean, \scale)\) represents the Laplace distribution with mean $\mean$ and scaling parameter $\scale$. A \emph{state} $\pstate$ is then a triple $\pstate = (\dompartialfunc, \randompartialfunc, G)$, where $\dompartialfunc: \cX \partialfunc \sdom$, $\randompartialfunc: \cR \partialfunc \distr \cup \expr$, and \(G \subseteq \BExp\). Here $\dompartialfunc$ is a partial map that assigns a value to $\sdom$ variables, $\randompartialfunc$ is a partial function mapping real variables to either the distribution from which they are sampled or the expression that is assigned to them, and $G$ is a set of Boolean conditions.

For a state $\pstate$ and a variable $\rv \in \cR$, we say that $\rv$ is an \emph{independent} variable if $\randompartialfunc(\rv) \in \distr$ and a \emph{dependent} variable if $\randompartialfunc(\rv) \in \expr$. For a state $\pstate$, ($\sdom$ or real) variable $\mathsf{\bf v}$ and value $u \in \sdom \cup \distr \cup \expr$, $\pstate[\mathsf{\bf v} \mapsto u]$ denotes the state that maps $\mathsf{\bf v}$ to $u$ and is otherwise identical to $\pstate$.

\paragraph{Final States}
$\executions{\Prog}{\pstate}$ denotes the set of \emph{final states} reached when $\Prog$ is executed from starting state $\pstate$. It is inductively defined as follows.
\begin{itemize}
  \item \(\executions{\grammarskip}{\pstate} = \{\pstate\}\).

   \item \(\executions{\dv \gets d}{\pstate} = \{\pstate[\dv \mapsto d]\}\).
    

    \item \(\executions{\rv \gets \gauss{\dv, \frac{a}{\epsilon}}}{\pstate} = \left\{\pstate\left[\rv \mapsto \left(\gaussian, \pstate(\dv), \frac{a}{\epsilon}\right)\right]\right\}\).


  \item \(\executions{\rv \gets \lap{\dv, \frac{a}{\epsilon}}}{\pstate} = \left\{\pstate\left[\rv \mapsto \left(\mylaplace, \pstate(\dv), \frac{a}{\epsilon}\right)\right]\right\}\).


    \item  \(\executions{\rv \gets R}{\pstate} = \left\{\pstate\left[\rv \mapsto R' \right]\right\}\), where \( R' \) is the expression obtained by replacing every dependent \( \rv' \in \mathcal{R} \) that appears in \( R \) with \(\randompartialfunc(\rv') \).

   \item
    \(
    \executions{\ifstatement{B}{\Prog_1}{\Prog_2}}{\pstate} = \executions{\Prog_1}{\pstate_{\{B\}}} \cup \executions{\Prog_2}{\pstate_{\{\neg B\}}}
    \),
    where $\pstate_{\{B\}} = \pstate[G \mapsto \pstate(G) \cup \{B\}]$ and $\pstate_{\{B\}} = \pstate[G \mapsto \pstate(G) \cup \{\neg B\}]$. Here, \( \neg B \) denotes the ``flipped'' comparison: for $\sim \in \set{<,\leq,>,\geq,=,\neq}$, $\neg (R \sim R') \triangleq (R \overline{\sim} R')$, where $\overline{<} = \geq$, $\overline{\leq} = >$, $\overline{>} = \leq$, $\overline{\geq} = <$, $\overline{=} = \neq$, and $\overline{\neq} = =$.
      

  \item \(
    \executions{\Prog_1; \Prog_2}{\pstate} = \smashoperator{\bigcup_{\mathclap{\tau \in \executions{\Prog_1}{\pstate}}}} \{\tau' \mid \tau' \in \executions{\Prog_2}{\tau}\}.
    \)
\end{itemize}





\paragraph{Input/Output Behavior.}
Let us fix a program $\Prog = (\inputset, \outputset, S)$, an input valuation $u: \inputset \to \sdom$ and output valuation $o: \outputset \to \sdom$. The \emph{final states of $\Prog$ on input $u$ with output $o$}, denoted $\run(u,o,\Prog)$, is defined as 
\[
\run(u,o,\Prog) = \set{(\dompartialfunc, \randompartialfunc,G) \in \executions{\Prog}{\pstate_0} \mid \forall y \in \outputset,\dompartialfunc(y) = o(y)}
\]
where the initial state $\pstate_0 = (\dompartialfunc_0,\randompartialfunc_0,G_0)$ has $G_0 = \emptyset$, $\randompartialfunc_0$ is the partial function with empty domain, and $\dompartialfunc_0$ is the partial function with domain $\inputset$ such that $\dompartialfunc_0(x) = u(x)$ for $x \in \inputset$.

\paragraph{Probability of a state.}
Let $\tau = (\dompartialfunc,\randompartialfunc,G)$ be a state. Define
\[
\Gconst = \{ g \in G \mid g = \dv \sim \dv', \text{ where } \dv,\dv'\in\cX \},
\]
and
\[
\Grandom = \{g \in G \mid g = \rv\sim \dv \text{ or } g = \rv \sim \rv' \text{ where } \rv,\rv'\in\cR, \dv \in \cX\}.
\]
Let $\evalconstant(\Gconst)$ be the Boolean value given by
\begin{align*}
\evalconstant(\Gconst) = \bigwedge_{(\dv \sim \dv') \in \Gconst} \dompartialfunc(\dv) \sim \dompartialfunc(\dv').
\end{align*}
In the above equation, when $\Gconst = \emptyset$, the conjunction is taken to be $\true$ as is standard.


Let \(\rand\) be the partial function on \(\cR\) with the same domain as $\randompartialfunc$ defined as follows. If
\(
\randompartialfunc(\rv) = (\gaussian, \mu, \sigma),
\)
then \(\rand(\rv)\) is the Gaussian random variable \(\randomvar_\rv\) with parameters \((\mu, \sigma)\). If
\(
\randompartialfunc(\rv) = (\mylaplace, \mu, b),
\)
then \(\rand(\rv)\) is the Laplace random variable \(\randomvar_\rv\) with parameters \((\mu, b)\). If $\randompartialfunc(\rv)=R$, where $R\in \expr$, then \(\rand(\rv)\) is the expression obtained from $R$ by replacing every independent variable $\rv'$ appearing in $R$ by the random variable \(\randomvar_{\rv'}\).
Now, let us define
\begin{align*}
    \evalrandom(\Grandom) = & \left(\bigwedge_{\rv \sim \rv' \in \Grandom:\ \rv,\rv' \in \cR} \rand(\rv) \sim \rand(\rv')\right) \wedge \\
     & \left(\bigwedge_{\rv \sim \dv \in \Grandom:\ \rv \in \cR,\ \dv \in \cX} \rand(\rv) \sim \dompartialfunc(\dv)\right)
\end{align*}

Now, for a given value to the parameter $\epsilon$, 
the probability of $\tau$ is given by:
\[
\pr{\epsilon,\tau} =
\begin{cases}
0 & \text{if } \evalconstant(\Gconst) = \false \\
\pr{\epsilon,\evalrandom(\Grandom)} &\text{otherwise} 
\end{cases}
\]
where $\pr{(\epsilon,\evalrandom(\Grandom))}$ is the probability that the random variables $\randomvar_\rv$ satisfy the condition $\evalrandom(\Grandom)$ for the given value of $\epsilon.$

\paragraph{Probability of Output.}
For any given $\epsilon>0$, the probability that the program $\Prog$ outputs the valuation $o$, with input values given by valuation $u$, denoted by $\pr{\epsilon,u,o,\Prog}$, is defined as
\[\pr{\epsilon,  u, o,\Prog} = \sum_{\tau \in \run(u,o,\Prog)} \pr{\epsilon, \tau}\]

\begin{example}
\label{example:svtfinalstate}
For the {\SVTGauss} program given in Algorithm~\ref{fig:SVTG}, called $\Prog$ here, the set of input variables  $\inputset = \{q_1, q_2\}$, and the set of output variables  $\outputset = \{out_1, out_2\}$.  

Consider an input assignment $u = \{q_1 \mapsto 0, q_2 \mapsto 1\}$ and an output assignment $o = \{out_1 \mapsto 0, out_2 \mapsto 1\}$. In this case, it can be easily seen that there is a single state $\tau\:=(\dompartialfunc,\randompartialfunc,G)$ in $\run(u,o,\Prog)$ where 
$\dompartialfunc(out_1)=0$, $\dompartialfunc(out_2)=1$ and $G=\Grandom=\{\rv_1 < \rv_T, \rv_2 \geq \rv_T\}.$

Now, we have 
\[\pr{\epsilon,  u, o,\Prog} = \pr{\epsilon, \tau} = \pr{X_{\rv_1} < X_{\rv_T} \wedge X_{\rv_2} \geq X_{\rv_T}}\]

\end{example}

For a set of output valuations $F$, the probability of $P$ producing an output in $F$ on input $u$ can be defined as $\sum_{o \in F} \pr{\epsilon,u,o,\Prog}$. Using this, the definitions $(\Depsilon,\delta)$-differential privacy and $(\Depsilon,\delta)$-strict differential privacy given in Definition~\ref{def:differential-privacy}, can be precisely instantiated for {\ourlang} programs.

\rmv{
\begin{definition}
\label{def:diff-priv-program}
\brown{A {\ourlang} $\Prog_{\epsilon}$ is  $(\Depsilon, \delta)$-differentially private (for $\epsilon>0$, $\Depsilon > 0$ and $\delta\in \unit$ with respect to $\adjacent$) iff for all subsets $F\subseteq \cV$ and
$u,u' \in \cU$ such that $(u,u')\in \Phi$,}
\begin{equation}
\label{eqDP}
\sum_{o \in F} \pr{\epsilon,u,o,\Prog} \leq 
\eulerv{\Depsilon}
\sum_{o \in F} \pr{\epsilon,u',o,\Prog} + \delta
\end{equation}
\end{definition}
\rc{Define strict differential privacy}
}

\begin{definition}
    The \emph{differential privacy problem} is the following: Given a {\ourlang} program $\Prog$, adjacency relation $\adjacent$ on inputs of $\Prog,$ rational numbers $\epsilon_0>0$, $\Depsilon>0$, $\delta \in \unit$, determine if $\Prog$ with privacy parameter taking value $\epsilon_0$ is $(\Depsilon, \delta)$-differentially private with respect to $\adjacent$.
\end{definition}


%% file: sec_checking_dp.tex
We describe our core algorithms for checking differential privacy of programs. For the rest of the section, we assume that we are given a {\ourlang} program $\Prog=(\inputset,\outputset,S)$ with privacy parameter $\epsilon$. Let $\cU$ be the set of possible functions from $\inputset$ to $\sdom,$ and  $\cV$ be the set of possible functions from $\outputset$ to $\sdom$ respectively, 
$\Depsilon$ denote the privacy budget. Let $\delta$ denote the error parameter.  
\begin{definition}
    A precision $\precision$ is a natural number. A $\precision$-approximation of a real number $p$ is an interval $[L,U]$ such that $L,U$ are rational numbers, $L\leq p\leq U$ and $U-L\leq 2^{-\precision}.$
\end{definition}

Verifying differential privacy of $\Prog$ requires checking inequalities across all subsets of possible outputs, as specified in Definition~\ref{def:differential-privacy}. Two key challenges arise in this context. We describe the challenges below and develop two key innovations to tackle them.   
\begin{enumerate}
[labelindent=0.1\parindent,
  itemindent=0pt,
  leftmargin=*,]
    \item Apriori, the inequality in Definition~\ref{def:differential-privacy} needs to be checked exponentially many times as there are $2^{\abs\outputset}$ possible subsets of outputs.  For example, consider Algorithm~\ref{fig:SVT} with $N$ inputs: we have $N$ possible outputs, resulting in $2^{N}$ subsets to check for each adjacent pair. Given that we have to check \emph{all} adjacent pairs, this makes these checks even more expensive. 
    Instead of checking every possible subset, we rephrase the differential privacy definition so that only one equation needs to be checked for each adjacent input (See Lemma~\ref{lem:DiffPrivacyRephrased}).
    \item As mentioned in the Introduction, it is unclear that $\pr{\epsilon,u,o,\Prog}$ can be computed exactly. Hence, instead of computing $\pr{\epsilon,u,o,\Prog}$ exactly, we compute 
    $\precision$-approximations of $\pr{\epsilon,u,o,\Prog}$ and $\eulerv \Depsilon \pr{\epsilon,u,o,\Prog}$  for a given precision $\precision.$
    This allows us to design an algorithm, {\VerifyDP} that returns three possible values: $\diffprivate,$ $\notdiffprivate,$ and $\unresolved.$ The algorithm is sound in that if it returns $\diffprivate$ ($\notdiffprivate$), then the input program $\Prog$ is differentially private (not differentially private, respectively).  
\end{enumerate}

The following lemma whose proof is given in 
allows us to tame the exponential number of subsets of outputs in the differential privacy checks.  
\begin{lemma}
\label{lem:DiffPrivacyRephrased}
A {\ourlang} program $\Prog$ is  $(\Depsilon, \delta)$-differentially private (for $\Depsilon > 0$ and $\delta\in \unit$)  with respect to $\adjacent$ iff for all  $(u,u')\in \adjacent$,  
\begin{equation}
\label{eqDP2}
\deltaadjs{u}{u'}=\sum_{o \in \cV} \max(\pr{\epsilon,u,o,\Prog}- 
\eulerv{\Depsilon} \pr{\epsilon,u',o,\Prog},0) \leq  \delta
\end{equation}
\end{lemma}
\begin{proof}
($\Rightarrow$) Let $\Prog$ be  $(\Depsilon,\delta)$ differentially private and $u,u'\in\Phi.$ By setting $F = \{o \in \cV \mathbin{|} \pr{\epsilon,u,o,\Prog}-\eulerv{\Depsilon}\pr{\epsilon,u',o,\Prog} > 0\}$, we can conclude that  Equation (4) is true for $\epsilon,\Prog, u,u'$ as follows. 

$$\begin{array}{l}
\sum_{o \in  \cV} \max(\pr{\epsilon,u,o,\Prog} - 
\eulerv{\Depsilon} \pr{\epsilon,u',o,\Prog},0) \\
\hspace{0.4cm}= \sum_{o \in  F } \max(\pr{\epsilon,u,o,\Prog} - \eulerv{\Depsilon} \pr{\epsilon,u',o,\Prog},0)\\
\hspace{0.5cm} + \sum_{o \in  \cV\setminus F } \max(\pr{\epsilon,u,o,\Prog} -\eulerv{\Depsilon}\pr{\epsilon,u',o,\Prog},0).
\end{array}$$

For $o\in F$, we have that
$$\begin{array}{l}\max(\pr{\epsilon,u,o,\Prog} - 
\eulerv{\Depsilon}
 \pr{\epsilon,u',o,\Prog},0) \\=\pr{\epsilon,u,o,\Prog} - 
\eulerv{\Depsilon}
 \pr{\epsilon,u',o,\Prog}.\end{array}$$

For $o \notin F$, we have that  $$\max(\pr{\epsilon,u,o,\Prog} - 
\eulerv{\Depsilon}
 \pr{\epsilon,u',o,\Prog},0)=0.$$
 
Thus, $\sum_{o \in  \cV } \max(\pr{\epsilon,u,o,\Prog} - 
\eulerv{\Depsilon}\pr{\epsilon,u',o,\Prog},0) 
         = \sum_{o \in  {F} } \pr{\epsilon,u,o,\Prog} - 
\eulerv{\Depsilon}
 \pr{\epsilon,u',o,\Prog} \leq  \delta$ as $\Prog$ is $(\Depsilon,\delta)$ differentially private.
 
($\Leftarrow$) This direction follows from the observation that for each $u,u'$ and arbitrary $F\subseteq \cV,$
$$\begin{array}{l}\sum_{o \in F} \pr{\epsilon,u,o,\Prog} - \eulerv{\Depsilon}
\sum_{o \in F} \pr{\epsilon,u',o,\Prog}\\
 \hspace*{0.3cm}=\sum_{o \in F} (\pr{\epsilon,u,o,\Prog} - \eulerv{\Depsilon}
\pr{\epsilon,u',o,\Prog})\\
 \hspace*{0.3cm}\leq \sum_{o \in F} \max(\pr{\epsilon,u,o,\Prog} - 
\eulerv{\Depsilon}
 \pr{\epsilon,u',o,\Prog},0) \\
 \hspace*{0.3cm}\leq \sum_{o \in  \cV } \max(\pr{\epsilon,u,o,\Prog} - 
\eulerv{\Depsilon}
 \pr{\epsilon,u',o,\Prog},0).
 \end{array}$$

The second last line of the above sequence of inequalities follow from the fact that for any real number $a$, $a\leq \max(a,b).$ The last line follows from the fact that for all $o\not\in F, \max(\pr{\epsilon,u,o,\Prog} - 
\eulerv{\Depsilon}  \pr{\epsilon,u',o,\Prog},0) \geq 0.$
\end{proof}
It will be useful to consider the set of error parameters for which Equation~\ref{eqDP2} becomes an equality.
\begin{definition}
\label{def:critical}
For a program $\ProgEps$ with adjacency relation $\adjacent,$ privacy budget $\Depsilon>0$ and error parameter $\delta \in \unit,$  the set of \emph{critical} error parameters is defined to be the set
\begin{align*}\critical=\set{ \deltaadjs{u}{u'}  \st (u,u')\in\adjacent \, \mathbin{\&}\,  &
 \deltaadjs{u}{u'}=\sum_{o\in \cV} \mathsf{max}(\pr{\epsilon,u,o, \Prog}
 \\& -\eulerv{\Depsilon} \pr{\epsilon,u',o,\Prog},0) }.
\end{align*}
\end{definition}
We shall now describe the {\VerifyDP} algorithm that allows us to verify (soundly) differential and non-differential privacy. 

\subsection[VerifyDP Algorithm]{$\textsf{VerifyDP}$ algorithm}
We  will assume that we can approximate $\pr{\epsilon,u,o,\Prog}$ and $\eulerv\Depsilon \pr{\epsilon,u,o,\Prog}$ to any desired degree of precision. We shall refer to $\pr{\epsilon,u,o,\Prog}$ as output probability and  $\eulerv\Depsilon \pr{\epsilon,u,o,\Prog}$ as scaled output probability.
\begin{definition}
We say that the output probability is effectively approximable if there is an algorithm $\compute(\cdot)$ such that  on input $\precision$, rational number $\epsilon>0,$ $u\in \cU, o\in \cV,$ and {\newlanggauss} program $\Prog$, $\compute(\epsilon,u,o,\Prog)$  outputs a $\precision$-approximation of 
$\pr{\epsilon,u,o,\Prog}$.

We say that the scaled output probability is effectively approximable if there $\computescale(\cdot)$ such that  on input $\precision$, rational $\Depsilon,\epsilon>0$, $u\in \cU,o\in \cV$ and {\newlanggauss} program $\Prog,$  $\computescale(\Depsilon,\epsilon,u,o,\Prog,\Depsilon)$ outputs a $\precision$-approximation of 
$\eulerv\Depsilon\pr{\epsilon,u,o,\Prog}$, respectively.


If the algorithm $\compute(\cdot)$ ($\computescale(\cdot)$, respectively) computes $\precision$-approximation for $\pr{\epsilon,u,o,\Prog}$ ($\eulerv\Depsilon\pr{\epsilon,u,o,\Prog}$, respectively) for a specific $\precision$ only (and not for all $\precision$), 
we say that the the output probability (scaled output probability, respectively) is effectively $\precision$-approximable.
\end{definition}



We have the following.
\begin{proposition}
\label{prop:approximable}
    If the output probability is effectively approximable then the scaled output probability is effectively approximable. 
\end{proposition}



The $\VerifyDP$ algorithm (See Algorithm \ref{alg:overall-verification}) checks differential privacy for $\Prog=(\inputset,\outputset,S)$ for all input pairs given by the adjacency relation \( \adjacent \), privacy parameter \( \epsilon \), error parameter \( \delta \) and privacy budget $\Depsilon$. The algorithm assumes the existence of  $\compute(\cdot),$ $\computescale(\cdot),$ and proceeds as follows.

\begin{algorithm}[ht]
\DontPrintSemicolon

\KwIn{Program $\Prog$,  Adjacency $\adjacent$,   Privacy parameter $\epsilon$,  Error parameter $\delta$, Privacy Budget $\Depsilon$, precision $\precision$}
\;
\KwOut{One of
(a) $\diffprivate$ (satisfies DP)
(b) $\notdiffprivate$ (violates DP)
(c) $\unresolved$
}
\;
$\store \gets \emptyset$\;
$b \gets \true$\;
\ForEach{$(u, u')$ $\in$ $\adjacent$}{
    \ForEach{$o$ $\in$ $\cV$}{
        \If{$(o, u) \notin \store$}{
            $\store[(o, u)] \gets \compute( \epsilon, u, o, \Prog)$\;
        }

        \If{$(o, u') \notin \stores$}{
            $\stores[(o, u')] \gets \computescale( \Depsilon,\epsilon, u, o, \Prog)$\;
        }
    }

    $\s{res} = \VerifyDPPair(u, u', \delta, \store,\stores)$\;

    \If{\textnormal{$\s{res}$ $=$ $\notdiffprivate$}}{
       \Return $\notdiffprivate$\;
    }

     \If{\textnormal{$\s{res}$ $=$ $\unresolved$}}{
       $b \gets \false$\;
    }
}

\If{b}{
    \Return $\diffprivate$\;
}
\Return $\unresolved$\;
\caption{\VerifyDP}
\label{alg:overall-verification}
\end{algorithm}

\begin{algorithm}[ht]
\SetAlgoLined
\DontPrintSemicolon
\SetKwProg{Fn}{Function}{}{}
\Fn{\VerifyDPFunc{$u, u', \delta, \store,\stores$}}{

\nl $\sumdeltamin \gets 0$\;
\nl $\sumdeltamax \gets 0$\;

\nl \ForEach{$o$ $\in$ $\cV$}{
\nl   $\intprob_{u} \gets \store[(o, u)]$\;
\nl   $\intprob_{u'} \gets \stores[(o, u')]$\;
\nl    $L_1, U_1 \gets \lowerbound(\intprob_{u}), \upperbound(\intprob_{u})$\;
\nl     $L_2, U_2 \gets \lowerbound(\intprob_{u'}), \upperbound(\intprob_{u'})$\;
\nl    $\sumdeltamax \gets \sumdeltamax + \max(U_1 - L_2,0)$\;
 \nl   $\sumdeltamin \gets \sumdeltamin + \max(L_1 - U_2,0)$\;
 }
\nl \If{$\sumdeltamax \leq \delta$}{
 \nl   \Return $\diffprivate$\;
}

\nl \If{$\sumdeltamin > \delta$}{
 \nl  \Return $\notdiffprivate$\;
}
\nl \Return $\unresolved$\;

}
\caption{\VerifyDPPair}
\label{alg:abstract-dp}
\end{algorithm}

The algorithm processes one adjacent pair at a time.  The algorithm also maintains a flag $b$. Intuitively, the flag $b$ is $\true$ if $\Prog$ is differentially private for all input pairs $u$ and $u'$ checked thus far. 
For each adjacent pair $(u,u') \in \adjacent$, and each output $o\in\cV,$ the algorithm calls $\compute(\epsilon,u,o,\Prog),$   and $\computescale(\Depsilon,\epsilon,u',o,\Prog).$  The resulting values are
 stored in dictionaries $\store$ and $\stores.$ The dictionary $\store$ stores the output probabilities, and the dictionary $\stores$ stores the scaled output probabilities.
Once the probabilities for each output $o$ have been stored, the algorithm calls the function $\VerifyDPPair$ to either prove or disprove differential privacy for the input pair $u$ and $u'$, or to indicate that the current precision $\precision$ is insufficient (i.e., the result is $\unresolved$).

If $\VerifyDPPair$ returns $\notdiffprivate$ for any input pair, the flag $b$  is set to $\false$ and the algorithm  terminates immediately, concluding that the program is not differentially private. In cases where $\VerifyDPPair$ returns $\unresolved$ for a particular pair, the flag $b$  is set to $\false$. However, we continue to process the remaining pairs in $\adjacent.$ The rationale is that we may be able to conclude that $\Prog$ is not differentially private for some other pair that has not been checked as yet.   

After processing all adjacent pairs, if the flag $b$ remains $\true$, the algorithm returns that the program is differentially private; otherwise, it returns $\unresolved$.

\subsection*{The \(\VerifyDPPair\) function}
We now discuss the \(\VerifyDPPair\) function (Algorithm~\ref{alg:abstract-dp}). Given an adjacent input pair $u,u'$,  the error parameter \(\delta\), the dictionaries \(\store\) and \(\stores\), and precision $\precision$, \(\VerifyDPPair\) is tasked with proving or disproving differential privacy.

The function iterates over the set of outputs. For each output $o,$ it computes an upper bound and lower bound on $\max(\pr{\epsilon,u,o,\Prog}- 
\eulerv{\Depsilon} \pr{\epsilon,u',o,\Prog},0)$ as follows.  
For each output $o$, it retrieves that intervals $\intprob_u=[L_1,U_1]$ and $\intprob_{u'}=[L_2,U_2]$, the intervals containing the  $\pr{\epsilon,u,o,\Prog}$ and   $\eulerv\Depsilon\pr{\epsilon,u',o,\Prog}$ respectively.  Now, $\max(U_1-L_2,0)$ ($\max(L_1-U_2,0)$, respectively) can be seen to be an upper bound (lower bound, respectively) on $\max(\pr{\epsilon,u,o,\Prog}- 
\eulerv{\Depsilon} \pr{\epsilon,u',o,\Prog},0).$

The upper bounds and lower bounds on $\max(\pr{\epsilon,u,o,\Prog}- 
\eulerv{\Depsilon} \pr{\epsilon,u',o,\Prog},0)$ are accumulated in $\sumdeltamax$ and
$\sumdeltamin$ respectively. Once the iteration over the set of outputs is over, $\VerifyDPPair$ declares that $\Prog$ is differentially private for $u,u'$ if $\sumdeltamax\leq \delta$ and not differentially private if $\sumdeltamin > \delta.$
If neither $\sumdeltamax\leq \delta$ not $\sumdeltamin > \delta$, then $\VerifyDPPair$ returns  $\unresolved.$

\subsection[On the soundness and completeness of VerifyDP]{On the soundness and completeness of \(\VerifyDP\)} 

The following lemma states that \(\VerifyDP\) always gives a sound answer. 
 We will postpone the proof of the Lemma, and prove it along with Lemma~\ref{lem:DpVerify2}.
\begin{lemma}[Soundness of \(\VerifyDP\)]
\label{lem:DpVerify1}
Given precision $\precision,$
let the output and scaled output probabilities be effectively $\precision$-approximable.
Let $\Prog=(\inputset,\outputset,S)$ be a program with privacy parameter $\epsilon.$ Let $\adjacent$ be an adjacency relation, $\Depsilon>0$ be a privacy budget and $\delta\in \unit$ be an error parameter.
\begin{enumerate}
\item If $\VerifyDP(\Prog,\adjacent,\epsilon,\Depsilon,\delta,\epsilon)$   returns $\notdiffprivate$ for precision $\precision$   then $\Prog$ does not satisfy  $(\Depsilon,\delta)$-differential privacy with respect to $\adjacent$. 
\item If $\VerifyDP(\Prog,\adjacent,\epsilon,\Depsilon,\delta,\epsilon)$   returns $\diffprivate$ for precision $\precision$  then $\Prog$ satisfies  $(\Depsilon,\delta)$-differential privacy with respect to $\adjacent$. 
\end{enumerate}
\end{lemma}

 

The following lemma states that, if $\Prog$ is not differentially private, then  \(\VerifyDP\) will return $\notdiffprivate$ for large enough precision, and if 
$\Prog$ is differentially private then  \(\VerifyDP\) will return $\diffprivate$ for all non-critical error parameters, for large enough precision. 
\begin{lemma} [Completeness of \(\VerifyDP\)]
\label{lem:DpVerify2}
Let the output probability be effectively approximable for all precision $\precision.$
Let $\Prog=(\inputset,\outputset,S)$ be a program with privacy parameter $\epsilon.$ Let $\adjacent$ be an adjacency relation, $\Depsilon>0$ be a privacy budget and $\delta\in \unit$ be an error parameter. 
\begin{enumerate}

\item If $\Prog$ does not satisfy $(\Depsilon,\delta)$-differential privacy with respect to $\adjacent$, then there is a precision $\precision_0$ such that $\VerifyDP(\Prog,\adjacent,\epsilon,\Depsilon,\delta,\epsilon)$ returns $\notdiffprivate$ for each $\precision>\precision_0.$ 
\item If $\Prog$ satisfies $(\Depsilon,\delta)$-differential privacy with respect to $\adjacent$, and $ \delta \not \in
\critical$ (see Definition~\ref{def:critical})
then there is a precision $\precision_0$ such that $\VerifyDP(\Prog,\adjacent,\epsilon,\Depsilon,\delta,\epsilon)$ returns $\diffprivate$ for each $\precision>\precision_0.$  
 

\end{enumerate}
\end{lemma}

\input{app_VerifyDP_proof}

Suppose, we run $\VerifyDP$ repeatedly for $\Prog$ by incrementing the precision $\precision$ until the algorithm returns {\diffprivate} or {\notdiffprivate}. If $\Prog$ is not differentially private then we are guaranteed to see {\notdiffprivate} by Lemma~\ref{lem:DpVerify2}. If   $\Prog$ is differentially private then we will eventually see {\diffprivate} for all non-critical error parameters $\delta$.  Thus, if we can show that if output probabilities are effectively approximable, 
we can conclude that checking non-differential privacy of {\newlanggauss} programs is recursively enumerable, and decidable for all but a finite set of error parameters. This is subject of next section.
  

%% file: app_VerifyDP_proof.tex
\begin{proof} (\textbf{Proofs of Lemma~\ref{lem:DpVerify1} and Lemma~\ref{lem:DpVerify2}})

We prove both Lemma~\ref{lem:DpVerify1} and Lemma~\ref{lem:DpVerify2} together.
Let $(u,u')\in \adjacent$. Let 
$$\deltaadjs{u}{u'}= \sum_{o\in \cV} \mathsf{max}(\pr{\epsilon,u,o, \Prog}-\eulerv{\Depsilon} \pr{\epsilon,u',o,\Prog},0).$$

Let $o\in \cV$ be an output.
After $\VerifyDP$ processes $(u,u')$ and the output $o$, let $L_1(u,o), U_1(u,o), L_2(u',o)$, and $U_2(u',o)$ be such that 
 \begin{align*}
     &\store[(u,o)]=[L_1(u,o),U_2(u,o)] \mbox{ and}\\
     &\stores[(u',o)]=[L_2(u',o),U_2(u',o)].\
 \end{align*}

We have the following equations:
\begin{align*}
&L_1(u,o)\leq \pr{\epsilon,u,o,\Prog} \leq  L_1(u,o) + \frac 1 {2^\precision}\\
&U_1(u,o)\geq \pr{\epsilon,u,o,\Prog} \geq  U_1(u,o) - \frac 1 {2^\precision}\\
&L_2(u',o)\leq \eulerv \Depsilon \pr{\epsilon,u',o,\Prog} \leq  L_2(u,o) + \frac 1 {2^\precision}\\
&U_2(u',o)\geq \eulerv \Depsilon \pr{\epsilon,u',o,\Prog} \geq  U_2(u,o) - \frac 1 {2^\precision}.\\
\end{align*}
This implies that 
\begin{align*}
L_1(u,o)- U_2(u',o) \leq \pr{\epsilon,u,o,\Prog} - & \eulerv \Depsilon \pr{\epsilon,u',o,\Prog} \\ &\leq L_1(u,o)- U_2(u',o) + \frac 2 {2^\precision} 
\end{align*}
and 
\begin{align*}
U_1(u,o)- L_2(u',o) \geq \pr{\epsilon,u,o,\Prog} - &  \eulerv\Depsilon \pr{\epsilon,u',o,\Prog} \\ &\geq U_1(u,o)- L_2(u',o) - \frac 2 {2^\precision}.
\end{align*}
Therefore, when $\VerifyDPPair(u,u',\delta,\store,\stores)$ is executed then at the end of the for loop in $\VerifyDPPair$, the following equations hold.\: 
\begin{align}
    &\sumdeltamin \leq \deltaadjs{u}{u'}  \leq \sumdeltamin + \frac {2 \abs{\cV}}
    {2^\precision}    \label{LemmaDP:eq1}\\
    &\sumdeltamax \geq \deltaadjs{u}{u'}  \geq \sumdeltamax - \frac {2 \abs{\cV}}
    {2^\precision}. \label{LemmaDP:eq2}
\end{align}

\paragraph*{Finishing Lemma~\ref{lem:DpVerify1} Proof.}
It is easy to see that Equation~\ref{LemmaDP:eq1} and Equation~\ref{LemmaDP:eq2} imply the  two parts of the of Lemma~\ref{lem:DpVerify1}.

\paragraph*{Finishing Lemma~\ref{lem:DpVerify2} Proof.}
Recall that if output probabilities are effectively approximable, then so are scaled output probabilities. (Proposition~\ref{prop:approximable}).
For part one of the Lemma~\ref{lem:DpVerify2}, observe that Equation~\ref{LemmaDP:eq1} implies that 
$$\sumdeltamin \geq \deltaadjs{u}{u'}  -\frac {2\abs \cV} {2^\precision}$$ and thus 
$$\sumdeltamin - \delta\geq (\deltaadjs{u}{u'}  - \delta) -\frac{2\abs \cV} {2^\precision}.$$

Now, part one of the Lemma~\ref{lem:DpVerify2} follows from the observation that if  $\Prog$ is not differentially private then there must be $(u,u')\in \adjacent$ such that $\deltaadjs{u}{u'}  - \delta>0$ and hence there is a precision $\precision_0$ such $(\deltaadjs{u}{u'}  - \delta) -\frac {2\abs \cV} {2^{\precision_0}} >0.$

For part two of the Lemma~\ref{lem:DpVerify2}, observe that we have from Equation~\ref{LemmaDP:eq2} that for each $(u,u')\in \adjacent,$
$$ -\sumdeltamax + \frac{2 \abs \cV}{2^\precision}\geq -\deltaadjs{u}{u'}  .$$
Hence, 
$$ \delta -\sumdeltamax  \geq (\delta -\deltaadjs{u}{u'} ) - \frac{2 \abs \cV}{2^\precision}.$$

Now if $\Prog$ is differentially private and $\delta\notin \critical,$ then  $\delta -\deltaadjs{u}{u'} >0$ for each $(u,u')\in \adjacent.$ As there are only a \emph{finite} number of pairs $(u,u')\in \adjacent,$ it implies that $$\min_{(u,u')\in \adjacent}(\delta-\deltaadjs{u}{u'} )>0.$$ 
From this, it is easy to see that there is a precision $\precision_0$ such that $(\delta -\deltaadjs{u}{u'} ) - \frac{2 \abs \cV}{2^{\precision_0}}>0$ for each 
$(u,u')\in \adjacent.$ Thus, $\VerifyDPPair$ will return $\diffprivate$ for each $(u,u')\in \adjacent$ when run with precision $\precision_0.$
\end{proof}

%% file: sec_construct_graph.tex
Assume that we are given a {\ourlang} program $\Prog_\epsilon=(\inputset,\outputset,S)$ with privacy parameter $\epsilon$. We describe the algorithm $\compute$ that computes $\precision$-approximants of $\pr{\epsilon,u,o,P}.$ Fix $\epsilon>0,u$ and $o$. 
Recall that $\pr{\epsilon,u,o,\Prog}$ (See Section \ref{sec:our_language}), is given by:
\[\pr{\epsilon,  u, o,\Prog} = \sum_{\tau \in \run(u,o,\Prog)} \pr{\epsilon, \tau}.\]
For computing the above value, we need to compute $\pr{\epsilon, \tau}$, for each final state $\tau\in \run(u,o,\Prog).$ Furthermore, observe that since the number of final states of a program is independent of the precision $\precision,$ it suffices to show that there is an algorithm that produces a $\precision_0$-approximant of $\pr{\epsilon, \tau}$ for given $\precision_0.$ 

Fix \( \tau = (\dompartialfunc, \randompartialfunc, G) \). As given in Section \ref{sec:our_language}, let \( \Gconst \) and \( \Grandom \) correspond to the set of guards in $G$ with comparison of domain variables and comparison of random variables, respectively. Suppose \( \evalconstant(\Gconst) = \true \). Then, $\pr{\epsilon, \tau}$ is given by \( \pr{\epsilon, \evalrandom(\Grandom)} \). If \( \evalconstant(\Gconst) = \false, \) $\pr{\epsilon, \tau}=0.$ Since, \( \evalconstant(\Gconst)\) can be easily computed, it suffices to show that \( \pr{\epsilon, \evalrandom(\Grandom)}\)  can be computed up-to any precision.

Recall that the variables appearing in the guards of $\Grandom$ are all independent variables in $\cR$. Let $\rv_0, \rv_1, \dots, \rv_{n-1}$ denote the variables that appear in the guards of $\Grandom$. 
Also, recall that $\randomvar_{\rv_0}, \randomvar_{\rv_1}, \dots, \randomvar_{\rv_{n-1}}$ are independent random variables of Laplacian or Gaussian distributions. Let $\mean_0, \mean_1, \dots, \mean_{n-1}$ denote the means of these random variables, and let $\std_0, \std_1, \dots, \std_{n-1}$ denote their standard deviations, respectively. 
Let $h_0, h_1, \dots, h_{n-1}$ be the corresponding probability density functions of these random variables. Observe that $\evalrandom(\Grandom)$ is a conjunction of linear constraints over $\randomvar_{\rv_0}, \randomvar_{\rv_1}, \dots, \randomvar_{\rv_{n-1}}$.

We shall exploit the observation that the probability $\pr{\epsilon, \tau}$ can be expressed as sum of nested definite integrals. The primary challenge in exploiting this observation is that such integrals may have  $\infty$ or $-\infty$ as bounds. 
 We will handle this challenge as follows.  Given a threshold $\thr>0$, we can write $\pr{\epsilon,\tau}$ as $\bpr{\epsilon,\tau,\thr} +  \tpr{\epsilon,\tau,\thr}$. Here $\bpr{\epsilon,\tau,\thr}=\displaystyle \pr{\epsilon, \evalrandom(\Grandom)\wedge \bigwedge_{0\leq i\leq n-1} (\mean_i-\thr \cdot  \std_i \leq  \randomvar_{\rv_i} \wedge  \set{\randomvar_{\rv_i} \leq \mean_i+\thr \cdot  \std_i )}}$ is the probability of obtaining output $o$ on input $u$ under the constraint that each sampled value $\randomvar_{\rv_i}$ in $\tau$ remains within $\thr \cdot  \std_i$ from $\mean_i$. The term $\tpr{\cdot}$ accounts for the tail probability, capturing cases where at least one sampled value $\randomvar_{\rv_i}$ deviates by at least $\thr \cdot  \std_i$ from $\mean_i$. 
 As we shall argue shortly, $\bpr{\cdot}$ can be computed to any desired precision.  Moreover, by known tail bounds,  we can always select $\thr$ such that $\tpr {\epsilon,\tau,\thr}$ arbitrarily small. This will guarantee that $\pr{\epsilon,\tau}$ can be computed to any required degree of precision. 
We have the following observation:
 \begin{lemma}[Choosing $\thr$]
 \label{lem:threshold}
 There is an algorithm that given program $\Prog,$  final state $\tau$ of $\Prog,$ rational number $\epsilon>0,$ and precision $\precision,$ outputs $\thr$ such that $0\leq \tpr {\epsilon,\tau,\thr}\leq \frac {1}{2 \cdot 2^{\precision}}.$
 \end{lemma}
\begin{proof}
 Given a threshold $\thr,$ it is easy to see that $0\leq \tpr {\epsilon,\tau,\thr}\leq \sum^{n-1}_{i=0}\tl{\thr,X_{\rv_i}, \mean_i, \sigma_i}.$ From the known tail bounds (See Section~\ref{sec:preliminaries}), we know that
 that, for each $i$, there is a monotonically decreasing function \(k_i(\cdot)\) such that \(\tl{\thr,X_{\rv_i}, \mean_i, \sigma_i}\leq k_i(\thr).\)
 Furthermore, $\lim_{\thr \to \infty} k_i(\thr)=0.$  From these observations and the known tail bounds,  
we can choose $\thr_i$ such that $0\leq \tl{\thr,X_{\rv_i}, \mean_i, \sigma_i}\leq \frac {1}{2n 2^{\precision}}$ for each $\thr\geq \thr_i.$
 The result follows by choosing $\thr=\max_{0\leq i\leq n-1} \thr_i. $
\end{proof}

\subsection{Computing probabilities via integral expressions}

The computation of  $\bpr{\epsilon,\tau,\thr}$ for a given threshold $\thr$ is accomplished by constructing a proper nested definite integral expression.\footnote{By a definite proper integral, we mean an integral where the function being integrated is continuous on a bounded finite interval, and both the limits of integration are finite.}  
To derive the integral expression, we analyze the set $\Grandom$ of guards along with the set of the equations $\set{\mean_i-\thr_i\cdot \std_i \leq \randomvar_{\rv_i} \st 0\leq i\leq n-1} \cup \set{\randomvar_{\rv_i} \leq \mean_i+\thr_i\cdot \std_i \st 0\leq i\leq n-1}$. 



An integral expression $\intexpr$ over $\randomvar_{\rv_0},\ldots,\randomvar_{\rv_n}$ is said to be in {\emph{normalized}} form if there exists a permutation $\pi(0),\pi(1),...,\pi(n-1)$ of the indices $0,...,n-1$ and rational constants  $\theta_0^{-},\theta_0{^+},$ and \emph{linear} functions $\theta_i^{-}(y_0,y_1,...,y_{i-1})$, $\theta_i^{+}(y_0,y_1,...,y_{i-1})$ in the variables $y_0,...,y_{i-1},$ for $1\leq i\leq n-1$, 
such that 
\begin{equation}
\intexpr =\int_{\theta_0^{-}}^{\theta_0^{+}}h_{\pi(0)}(y_0)\int_{\theta_1^{-}}^{\theta_1^{+}}h_{\pi(1)}(y_1)\dots
\int_{\theta_{n-1}^{-}}^{\theta_{n-1}^{+}}h_{\pi(n-1)}(y_{n-1}) dy_{n-1}...dy_{0}
\label{eq:normalized-form}
\end{equation}

\begin{lemma}
    \( \bpr{\epsilon,\tau,\thr}=\displaystyle \pr{\epsilon, \evalrandom(\Grandom)\wedge \bigwedge_{0\leq i\leq n-1} (\mean_i-\thr \cdot  \std_i \leq  \randomvar_{\rv_i} \wedge  \set{\randomvar_{\rv_i} \leq \mean_i+\thr \cdot  \std_i )}} \)
    is a finite sum of integral expressions in \emph{normalized} form over the random variables $\randomvar_{\rv_{0}},...,\randomvar_{\rv_{n-1}}.$ 
    \end{lemma}

The proof of the above lemma follows from the results of~\cite{Polytope} and uses the same approach as that used in the proof of Lemma 18 in ~\cite{bcjsv20}. We show that the output probabilities are effectively approximable:

\begin{theorem}
\label{thm:computeprob}
  The output probabilities are $\precision$-approximable. That is, there is an algorithm 
  $\compute(\epsilon,u,o,\Prog)$ that takes inputs $\precision,\epsilon,u,o,\Prog$ and returns a $\precision$-approximation of  $\pr{\epsilon,u,o,\Prog}.$
 
\end{theorem}

\begin{proof}
Thanks to Lemma~\ref{lem:threshold}, it suffices to show that we can compute a $(\precision+1)$-approximation $[L,U]$ of $\bpr{\epsilon,\tau,\thr}.$ From our previous arguments, it follows that
$\bpr{\epsilon,\tau,\thr}$ can be obtained as a finite sum of normalized integral expressions of the form $\intexpr$ shown above. All the constants and coefficients of the linear functions used as limits of integrals in the different summands  can be obtained algorithmically from the values of $\thr,\epsilon,u,o,\Prog.$ The result now follows from the fact that all the probability density functions in {\newlanggauss} are computable, and that the set of computable functions are closed under summation, definite proper integration, and composition~\cite{ko-real,int-dreal}.
\end{proof}

We get immediately from Lemma~\ref{lem:DpVerify1}, Lemma~\ref{lem:DpVerify2} and Theorem~\ref{thm:computeprob} that we can automatically check $(\Depsilon,\delta)$-differential privacy problem of a {\ourlang} program  $\ProgEps$ for all non-critical error parameters (See Definition~\ref{def:critical}).  
\begin{theorem}
\label{thm:decidability}
The problem of determining if a {\ourlang} program $\ProgEps$ is not-$(\Depsilon,\delta)$-differentially private with respect to adjacency relation $\adjacent$ for a given $\epsilon>0,\Depsilon>0,\delta\in \unit$ is recursively enumerable. 

The problem of checking if $\ProgEps$ is $(\Depsilon,\delta)$-differentially private with respect to adjacency relation $\adjacent$ is \emph{decidable} for all $\delta\in \unit$, \emph{except} those in the finite set $\critical.$
\end{theorem}
\subsection{Optimization of Integral Expressions}
We will now present a method for transforming an integral expression $\intexpr$ in normalized form into another equivalent integral expression $\intexpr'$ , so that the depth of nesting of $\intexpr'$ is minimal to enable efficient evaluation. Without loss of generality, consider the integral expression $\intexpr$ as given by the Equation \ref{eq:normalized-form}. To do the transformation of $\intexpr$, we define a directed graph $\dependencyGraph_{\intexpr}$, called the dependency graph of $\intexpr$, given by \( \dependencyGraph_{\intexpr} = (V, E) \), where $V\:=\{\rv_0,\dots,\rv_{n-1}\}$ and $E$ is the set of edges $(\rv_{\pi(\ell)},\rv_{\pi(j)})$ such that $y_\ell$ appears in either of the linear expressions  $\theta_j^{-},\theta_j^{+}$, for $0\leq j,\ell<n.$

We propose a method that takes the dependency graph $\graph_{\intexpr}$ of an integral expression as given in the Equation \ref{eq:normalized-form} and generates an equivalent nested integral that has a minimal depth of nesting. 
This method is given by the Algorithm~\ref{alg:optimization}.  This algorithm consists of a recursive function $\GenExpr$ that takes as input a sub-graph of the dependency graph $\graph_{\intexpr}$ of an integral expression as given in Equation \ref{eq:normalized-form}. Initially the algorithm is invoked with $\graph_{\intexpr}$ as the argument. The permutation $\pi$ used in the algorithm, in all recursive invocations, is the permutation $\pi$ that is used in Equation \ref{eq:normalized-form}. 

The algorithm works as follows. Initially, it checks whether the graph \( \graph \) has multiple Weakly Connected Components, in short WCCs~\footnote{A weakly connected component in a directed graph is a maximal sub-graph such that the undirected version of the sub-graph, obtained by replacing each directed edge by an undirected edge, is connected.}. If $\dependencyGraph$ has more than one WCC, the algorithm is invoked recursively for each WCC, and the product of the resulting expressions is returned. Observe that the depth of nesting of the resulting integral expression is the maximum of the depths of integral expressions for the individual WCCs. 
Next, if the graph \( \graph \) consists of a single node \( \rv_{\pi(j)} \), the algorithm returns the integral expression of \( \rv_{\pi(j)} \). In the algorithm, for any node $\rv_{\pi(j)}$, $h_{\pi(j)}$ is the density function of the random variable $\randomvar_{\rv_{\pi(j)}}$ and the bounds of the integral are $\theta^{-}_j$ and $\theta^{+}_j$ which are given in Equation \ref{eq:normalized-form}.
Otherwise, it identifies the source nodes \( \sourcenodes \). A node is a source node if it has  no incoming edge. The algorithm then constructs nested integrals corresponding to source nodes sequentially, and invokes the algorithm on the reduced graph obtained by deleting these source nodes new graph \( \graph \setminus \sourcenodes \). We discuss the impact of our optimization in the experiments presented in Table~\ref{table:optimization}. Theorem \ref{thm:optimized}, given below, states the correctness of Algorithm~\ref{alg:optimization}.

\begin{theorem}
\label{thm:optimized}
For any normalized integral expression $\intexpr$ over $\randomvar_{\rv_1},\ldots,\randomvar_{\rv_n}$, as given by equation \ref{eq:normalized-form}, and with $\graph_{\intexpr}$ being the dependency graph of $\intexpr$, if $\intexpr'$ is the integral expression returned by $\GenExpr(\graph_{\intexpr})$, then $\intexpr'$ is equivalent to $\intexpr$; that is, for all probability density functions $h_0,...,h_{n-1}$, $\intexpr$ and $\intexpr'$ have the same values.
\end{theorem}
\begin{proof}
  The function $\GenExpr (\graph)$, given by Algorithm~\ref{alg:optimization} is a recursive function. It takes the dependency graph $\graph$ of an integral expression $E$ in normalized form as an argument and outputs another integral expression, which we show to be  equivalent to $E$, i.e., has the same value as $E$ for any given values to the free variables in $E$. The proof of the equivalence of $E$ and the expression returned by $\GenExpr(E)$ is sketched below. Observe that $E$ will not have free variables during the first invocation of $\GenExpr(\cdot)$ (i.e., when $E\:=\intexpr$), but in the subsequent recursive invocations the expression $E$ (corresponding to the argument $\graph$ of $\GenExpr(\cdot)$),  may have free variables which appear in the limits of the integrals appearing in $E$.

Each node in $\graph$ represents a random variable and an edge $(u,v)$ in $\graph$ indicates that the variable $u$ appears in the lower or upper limit of the integral corresponding to the variable $v$. 

Consider the case when $\graph$ has more than one weakly connected components. Specifically, assume  $\graph$ has two weakly connected components $\graph_1,\graph_2$. The integrals corresponding to the variables in $\graph_1$ can be moved leftwards to the front retaining their order of occurrence in $E$, resulting in an expression which is a product of two expressions $E_1$ (corresponding to $\graph_1$) and $E_2$ ( corresponding to $\graph_2$). The resulting product expression is equivalent to $E$ since none of the variables in $\graph_1$ depend on those in $\graph_2$. This argument generalizes when $\graph$ has more than two weakly connected components. This argument also holds for the recursive invocation $\GenExpr(\graph')$ in the return statement of the function.
The proof of the theorem follows from the above observations.
\end{proof}

\begin{algorithm}[htbp]

\SetAlgoLined
\DontPrintSemicolon
\KwIn{$\graph$ a Directed Acyclic Graph (DAG)}
\KwOut{Integral Expressions}
\SetKwProg{Fn}{Function}{:}{}
\Fn{\GenExprFunc{$\graph$}}{
\If{$\graph \textnormal{ has WCCs } \graph_1,{\scriptstyle\cdots}, \graph_n$}{
\Return $(\GenExpr(\graph_1))(\GenExpr(\graph_2)){\scriptstyle{\cdots}}(\GenExpr(\graph_n))$\;
}

\If{$\graph \textnormal{ is singleton node } \rv_{\pi(j)} \textnormal{ for some } j,0\leq j<k$}{
     \Return$\displaystyle\int_{\theta_{j}^{-}}^{\theta_{j}^{+}} h_{\pi(j)}(y_j)\,dy_j$\;
}
Let $\sourcenodes = \{\rv_{\pi(j_{0})},{\scriptstyle{\cdots}},\rv_{\pi(j_{\ell-1})}\},0<\ell\leq n$ be source nodes\;

\Return 
${\displaystyle\int_{\theta^{-}_{j_{0}}}^{\theta^{+}_{j_{0}}}\!\!\!\!\!\!{\scriptstyle{\cdots}}\int_{\theta^{-}_{j_{\ell-1}}}^{\theta^{+}_{j_{\ell -1}}} \prod_{i=0}^{\ell-1} h_{\pi(j_{i})}\,\GenExpr(\graph')\,dy_{j_{\ell-1}}{\scriptstyle{\cdots}} dy_{j_{0}}}$\;

where $\graph' = \graph \setminus \sourcenodes$\;
}
\caption{Optimized Integral Expressions Generation}
\label{alg:optimization}
\end{algorithm}


\begin{example}
\label{example:optimization} 
Consider  Example~\ref{example:svtfinalstate} on Page \pageref{example:svtfinalstate}.
For the {\SVTGauss} program given in Algorithm~\ref{fig:SVTG}, there is only one final state $\tau$ corresponding to the output $[0,1]$ on input $[0,1],$ where
 $\tau\:=(\dompartialfunc,\randompartialfunc,G)$ in $\run(u,o,\Prog)$ where 
$\dompartialfunc(out_1)=0$, $\dompartialfunc(out_2)=1$ and $G=\Grandom=\{\rv_1 < \rv_T, \rv_2 \geq \rv_T\}.$


We can write $\bpr{\epsilon,\tau,\thr}$ as the expression \[
 \intexpr  =\!\!\int_{-\thr\cdot \frac 2 \epsilon}^{ \thr \cdot\frac 2 \epsilon}\!\! f_{0, \frac{2}{ \epsilon}}(\rv_T) \! \int_{-\thr \cdot  \frac 4 \epsilon}^{\rv_T}\!\! f_{0,\frac{4}{ \epsilon}}(\rv_1) \! \int_{\rv_T}^{1+\thr \cdot \frac 4 \epsilon}\!\! f_{1, \frac{4}{\epsilon}}(\rv_2)\ d\rv_2 \!\ d\rv_1 \!\ d\rv_T.
\]
Figure~\ref{fig:modeling} shows the  dependency graph for $\intexpr.$
\begin{figure}[ht]
    \centering
        
\begin{tikzpicture}[->, >=Stealth, node distance=2cm,
    every node/.style={circle, draw, minimum size=0.6cm, inner sep=0pt}, scale=1.0]
    \node (r1) at (0,0) {$\rv_1$};
    \node (rT) at (0,0) [right of=r1] {$\rv_T$};
    \node (r2) [right of=rT] {$\rv_2$};
    \draw (rT) -- (r1);
    \draw (rT) -- (r2);
\end{tikzpicture}
\captionsetup{font=footnotesize,labelfont=footnotesize}  
\caption{\footnotesize Dependency graph of $\intexpr$ in Example~\ref{example:optimization}.}
\label{fig:modeling}
\end{figure}
Finally, the optimization algorithm (Algorithm~\ref{alg:optimization}) rewrites $\intexpr$ as
\[
 \!\!\int_{-\thr\cdot \frac 2 \epsilon}^{ \thr \cdot\frac 2 \epsilon}\!\! f_{0, \frac{2}{ \epsilon}}(\rv_T) \! \left(\int_{-\thr \cdot  \frac 4 \epsilon}^{\rv_T}\!\! f_{0,\frac{4}{ \epsilon}}(\rv_1) \!\ d\rv_1 \right) \left(\! \int_{\rv_T}^{1+\thr \cdot \frac 4 \epsilon}\!\! f_{1, \frac{4}{\epsilon}}(\rv_2)\ d\rv_2 \right)  \!\ d\rv_T.
\]
\end{example}

%% file: sec_experiments.tex
We implemented a simplified version of the algorithm, presented earlier, called the tool {\ourtool}. The tool is built using Python and C++ and is designed to handle {\ourlang} programs\footnote{Currently, {\ourtool} only supports comparison amongst sampled values.}, determining whether they are differentially private, not differentially private, or $\unresolved$. Given an input program $\Prog$ and an adjacency relation $\adjacent$, the tool checks differential privacy for fixed values of $\epsilon$, $\Depsilon$ and $\delta$.

{\ourtool} uses three libraries:  {\PLY}~\cite{PLY} for program parsing,  {\igraph}~\cite{igraph} for graph operations, and the {\FLINT} library~\cite{flint} for computing integral expressions. After parsing the program, we evaluate all final states of a program $\Prog$ (as given in Section~\ref{sec:our_language}). Afterwards, we  represent each final state as a graph, perform the ordering of integrals and compute their limits as described in Section~\ref{sec:compute_prob}. Once such integral expressions are generated, we encode them in C++, and use {\FLINT} to compute the interval probabilities of each final state for each input from the input pairs in the adjacency relation. Additionally, the tool refines the precision level to a higher level if the interval is too large to prove or disprove differential privacy. {\ourtool} is available for download at ~\cite{artifact}.

\subsection{Examples}
In this section, we present a limited set of examples from our benchmark suite due to space constraints. 
Details of the remaining examples and additional experimental results are provided in 
\ifdefined\short
\cite{ArxivPaper}.
\else
Appendix~\ref{app:psuedocode} and Appendix~\ref{app:fullexperimentalresults}.
\fi

\paragraph*{{\SVT} Variants.}
We categorized {\SVT} variants into three groups:  {\SVT} with Gaussian noise, {\SVT} with Laplace noise, and {\SVT} with mixed noise (where the threshold is sampled from a Laplace distribution and the queries from a Gaussian distribution, or vice versa). An example from the first category is {\SVTGauss} (Algorithm~\ref{fig:SVT}). 

\paragraph*{{\NoisyMax} and {\NoisyMin}.} In addition to the {\SVT} variants, we examine {\NoisyMax} and {\NoisyMin} with Gaussian or Laplacian noise. The Noisy Max with Gaussian algorithm 
selects the query with the highest value after independently adding Gaussian noise to each query result. This approach obscures the true maximum, ensuring differential privacy~\cite{DR14,DingWWZK18}.

\paragraph*{$k$-{\minmax} and $m$-{\Range}.} The $k$-{\minmax} algorithm (for $k \geq 2$) \cite{csvb23} perturbs the first $k$ queries, computes the noisy maximum and minimum, and then checks whether each subsequent noisy query falls within this range; if not, the algorithm exits. The $m$-{\Range} algorithm \cite{csvb23} perturbs $2m$ thresholds that define a rectangle of $m$ dimensions and checks whether noisy queries lie within these noisy limits. While the original algorithms use Laplace noise, we examine them also when the noise added is Gaussian. 


\subsection{Experiments}
We evaluated {\ourtool} on a macOS computer equipped with a 1.4 GHz Quad-Core Intel Core i5 processor and 8GB of RAM. Each example was executed three times, and the average execution time was recorded across these runs. Recall that when converting improper integrals with infinite limits into proper ones, we replace \( -\infty \) and \( \infty \) with tail bounds. We choose these bounds so that the remaining area in the tails is very small. We use $\thr=4$ for Gaussian distributions and  $\thr=8$ for Laplace distributions while computing these tail bounds and the error. We used an initial precision of $\precision = 16$ bits, which was refined up to 32 bits if differential privacy could not be verified. \emph{All pairs} for input size \( N \) refers to all pairs of input vectors in \(\{0,1\}^N\). Any pair of inputs in the \(\{0,1\}^N\) is adjacent. A \emph{single pair} of input size \( N \) refers to the pair of vector \( (0^N, 0^{N-1}1) \), where the first vector consists of \( N \) zeros and the second vector consists of \( N-1 \) zeros followed by a one.

\paragraph*{Performance and Scalability}
We experimented on variants of {\SVT} with input sizes from 1 to 26, on {\NoisyMax} and {\NoisyMin} with input sizes from 1 to 5, on $m$-{\Range} with $m = 2$ and input sizes from 1 to 3, and on $k$-{\minmax} with $k = 2$ and input sizes 3 and 4. {\ourtool} can verify whether a {\SVTGauss} variant is differentially private for a single input pair up to size $N = 25$. For larger inputs, it times out. Note that our timeout is set to 10 minutes; increasing this limit would allow {\ourtool} to handle more examples. As $N$ increases in the {\SVT} variants, we encounter out-of-memory error ($\experimenterror$) while storing probabilities for the case of all input pairs in the {\SVT} variants as the number of input combinations grows exponentially with $N$.
Performance results are presented in Table~\ref{table:examples} and Figure~\ref{fig:scalability} shows the scaling behavior with respect to input size $N$ and runtime for a single input pair.

\begin{figure}[tp]
    \centering
    \includegraphics[width=0.45\textwidth]{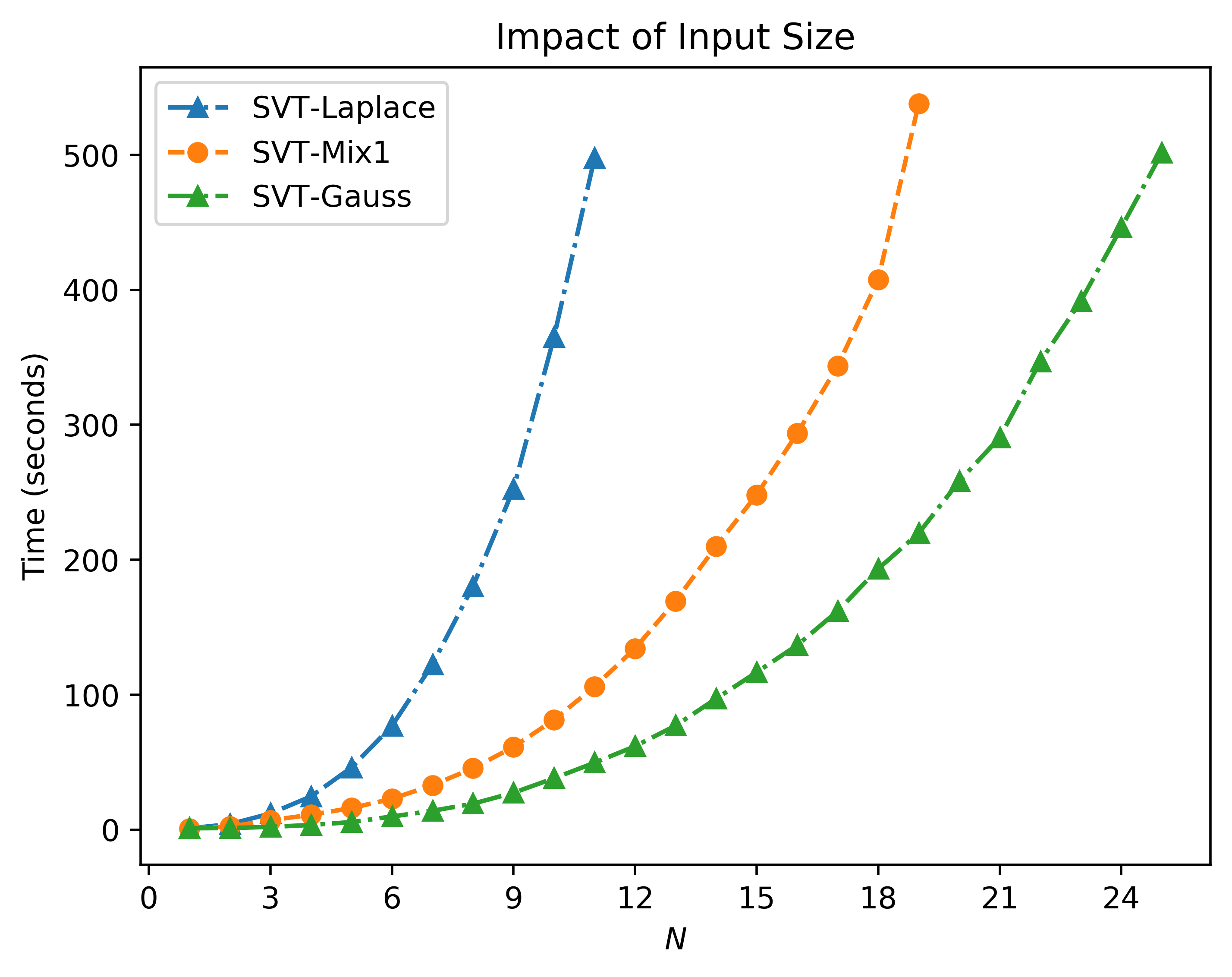}
    \captionsetup{font=footnotesize,labelfont=footnotesize}
    \caption{\footnotesize Scaling behavior of {\BelowThreshold}, {\LaplaceBelowThreshold}, and {\MixOneBelowThreshold} with varying input size $N$ for a single input pair.}
    \label{fig:scalability}
    
     \Description{Impact of Input Size on AboveThreshold, LaplaceAboveThreshold, MixOneAboveThreshold on the single pair.}
\end{figure}

\begin{table}[ht]
    \scriptsize
    \centering
    \renewcommand{\arraystretch}{1.5}
    \setlength\tabcolsep{4pt}
    \begin{tabular}{|lc|ccc|cc|cc|}
    \hline
        \multirow{2}{*}{Example} & \multirow{2}{*}{$N$} & Final & \multirow{2}{*}{$\overline{|G|}$} & Avg. &
        \multicolumn{2}{c|}{Single Pair} & \multicolumn{2}{c|}{All Pairs} \\\cline{6-9}
         
        & & States &  & Depth & DP? & Time (s) & DP? & Time (s) \\
        \hline
         \multirow{3}{*}{\BelowThreshold}  & 2 & 3 & 1.7 & 2.3 & $\checkmark$ & 1.6 & $\checkmark$ & 2.4 \\ 
         & 5 & 6 & 3.3 & 2.7 & $\checkmark$ & 7.9 & $\checkmark$ & 76.7 \\ 
         & 25 & 26 & 13.5 & 2.9 & $\checkmark$ & 441.3 & $\na$ & $\experimenterror$ \\ \hline

         \multirow{2}{*}{\NonPrivBelowThresholdOne}
         & 5 & 6 & 3.3 & 1.7 & $\times$ & 1.0 & $\times$ & 1.4 \\ 
         & 6 & 7 & 3.9 & 1.7 & $\times$ & 1.0 & $\times$ & 2.1 \\ \hline
         \multirow{2}{*}{\NonPrivBelowThresholdTwo}
         & 3 & 4 & 2.2 & 1.5 & $\times$ & 1.0 & $\times$ & 1.0 \\ 
         & 6 & 7 & 3.9 & 1.7 & $\times$ & 1.0 & $\times$ & 1.1 \\ \hline
         \multirow{3}{*}{\MixOneBelowThreshold}
         & 2 & 3 & 1.7 & 2.3 & $\checkmark$ & 2.4 & $\checkmark$ & 4.3 \\ 
         & 5 & 6 & 3.3 & 2.7 & $\checkmark$ & 19.5 & $\checkmark$ & 285.4 \\ 
         & 17 & 18 & 9.4 & 2.9 & $\checkmark$ & 365.7 & $\na$ & $\experimenterror$ \\ \hline
         \multirow{3}{*}{\NoisyMaxGauss}
         & 2 & 2 & 1.0 & 2.0 & $\checkmark$ & 1.0 & $\checkmark$ & 1.3 \\ 
        & 3 & 4 & 2.0 & 2.5 & $\checkmark$ & 1.6 & $\checkmark$ & 3.1 \\ 
        & 4 & 8 & 3.0 & 3.0 & $\checkmark$ & 37.8 & $\checkmark$ & 303.0 \\ \hline 
        \multirow{3}{*}{\NoisyMinGauss} & 2 & 2 & 1.0 & 2.0 & $\checkmark$ & 1.0 & $\checkmark$ & 1.3 \\ 
        & 3 & 4 & 2.0 & 2.5 & $\checkmark$ & 1.6 & $\checkmark$ & 3.1 \\ 
        & 4 & 8 & 3.0 & 3.0 & $\checkmark$ & 36.8 & $\checkmark$ & 303.6 \\ \hline
         \multirow{3}{*}{\mRange}
         & 1 & 7 & 3.0 & 2.5 & $\checkmark$ & 1.5 & $\checkmark$ & 1.8 \\ 
        & 2 & 13 & 4.2 & 3.2 & $\checkmark$ & 171.2 & $\checkmark$ & 344.4 \\ 
        & 3 & 19 & 5.2 & 3.8 & $\na$ & $\timeout$ & $\na$ & $\timeout$ \\ \hline
        \multirow{2}{*}{\kminmax}
         & 3 & 16 & 4.0 & 3.0 & $\checkmark$ & 2.1 & $\checkmark$ & 5.7 \\ 
         & 4 & 28 & 5.1 & 3.4 & $\checkmark$ & 41.2 & $\checkmark$ & 335.7 \\ 
        \hline
    \end{tabular}
    \captionsetup{font=footnotesize,labelfont=footnotesize}    \caption{\footnotesize Summary of Experimental Results for \ourtool. The columns in the table are defined as follows: $N$ is the input size of the program. DP? indicates whether the program is differentially private. \textbf{Final States} denotes the number of final states. $\overline{|G|}$ and Avg. Depth, respectively, denote the average number of conditions and the average nesting depth of integral expressions, per final state. Time is the average time (in seconds) to verify differential privacy, measured over three runs. $\timeout$ indicates a timeout (exceeding 10 minutes), and $\experimenterror$ denotes a run out of memory. Differential privacy checks were performed with $\epsilon = 0.5$ and $\delta = 0.01$, except for {\NonPrivBelowThresholdOne}, which uses $\epsilon = 8$. We used $\Depsilon = 0.5$, except for {\BelowThreshold} and {\MixOneBelowThreshold}, where $\Depsilon = 1.24$.}
    \label{table:examples}
\end{table}

\paragraph*{Impact of Optimization}
We conducted experiments to evaluate the impact of the optimized integral ordering presented in Algorithm~\ref{alg:optimization} on {\BelowThreshold}, {\LaplaceBelowThreshold}, {\MixOneBelowThreshold}, and {\NoisyMaxGauss} with input sizes ranging from $N=1$ to $N=5$. We compared the performance of the unoptimized version (which uses topological sorting for integral ordering) with the optimized version.
The results indicate that the optimized approach leads to a significant reduction in the maximum integration depth for all examples. These improvements also translate into substantial reductions in the overall running time. Table~\ref{table:optimization} summarizes the optimization results.

\begin{table}[htp]
    \centering
    \scriptsize
\setlength\tabcolsep{4pt}
\renewcommand{\arraystretch}{1.5}
    \begin{tabular}{|lc|ccc|ccc|}
    \hline
        \multirow{3}{*}{Example} & \multirow{3}{*}{$N$}  & \multicolumn{3}{c|}{Unoptimized} & \multicolumn{3}{c|}{Optimized} \\
        \cline{3-8}
        &  & Max & Avg. & \multirow{2}{*}{Time (s)} & Max & Avg. & \multirow{2}{*}{Time (s)}\\

        & & Depth & Depth & & Depth & Depth & \\
        \hline
        \multirow{5}{*}{\BelowThreshold}
        & 1 & 2 & 2.0 & 1.03 & 2 & 2.0 & 1.0 \\ 
        & 2 & 3 & 2.67 & 2.27 & 3 & 2.3 & 1.6 \\ 
        & 3 & 4 & 3.25 & $\timeout$ & 3 & 2.5 & 2.77 \\ 
        & 4 & 5 & 3.8 & $\timeout$ & 3 & 2.6 & 3.48 \\ 
        & 5 & 6 & 4.33 & $\timeout$ & 3 & 2.7 & 7.9 \\ \hline
        \multirow{5}{*}{\MixOneBelowThreshold}
         & 1 & 2 & 2.0 & 1.01 & 2 & 2.0 & 0.94 \\ 
        & 2 & 3 & 2.67 & 7.62 & 3 & 2.3 & 2.4 \\ 
         & 3 & 4 & 3.25 & $\timeout$ & 3 & 2.5 & 7.88 \\ 
        & 4 & 5 & 3.8 & $\timeout$ & 3 & 2.6 & 11.8 \\ 
        & 5 & 6 & 4.33 & $\timeout$ & 3 & 2.7 & 19.5 \\ \hline
        \multirow{4}{*}{\NoisyMaxGauss}
         & 2 & 2 & 2.0 & 0.98 & 2 & 2.0 & 1.0 \\ 
         & 3 & 3 & 3.0 & 3.53 & 3 & 2.5 & 1.6 \\ 
        & 4 & 4 & 4.0 & $\timeout$ & 4 & 3.0 & 37.8 \\ 
        & 5 & 5 & 5.0 & $\timeout$ & 5 & 3.5 & $\timeout$ \\ 
        \hline
    \end{tabular}
    \captionsetup{font=footnotesize,labelfont=footnotesize}\caption{\footnotesize Summary of optimization results for {\ourtool}. The columns in the table are as follows: $N$ represents the input size for the program. Time refers to the time taken to check differential privacy for a single pair, measured in seconds and averaged over three executions. $\timeout$ indicates a timeout (exceeding 10 minutes). Avg. Depth refers to the average nested depth of integrals across all executions. Max Depth refers to the maximum nested depth among all executions. Differential privacy checks were performed with $\epsilon = 0.5$ and $\delta = 0.01$. We used $\Depsilon = 1.24$ for {\BelowThreshold} and {\MixOneBelowThreshold}, and $\Depsilon = 0.5$ for {\NoisyMaxGauss}.}
    \label{table:optimization}
\end{table}

\subsection{Comparison with \dipc}
We compare the performance of our tool, {\ourtool}, with {\dipc}~\cite{bcjsv20}. We chose {\dipc} for comparison as it allows for verification of approximate differential privacy (and not just pure differential privacy), i.e., it allows for verifying $(\Depsilon,\delta)$ differential privacy for fixed values of $\epsilon,\Depsilon$ and $\delta.$ Like {\ourtool},
{\dipc} checks differential privacy for programs where both inputs and outputs are drawn from a finite domain and have bounded length. 
The key differences between {\dipc} and {\ourtool} are as follows: 
\begin{enumerate*}
    \item Unlike {\ourtool}, {\dipc}  does not support Gaussian distributions;
    \item {\dipc} can check pure differential privacy for all values of $\epsilon > 0$ as well as for fixed values of $\epsilon$. It also checks $(\Depsilon,\delta)$ differential privacy for fixed values of $\epsilon,\Depsilon$ and $\delta.$  {\ourtool}, on the other hand can only check for fixed values of $\Depsilon$ and $\delta$.
    \item {\dipc} relies on the proprietary software {\mathematica} for checking the encoded formulas, whereas {\ourtool} uses the open-source library {\FLINT}.
\end{enumerate*}

Our comparison is summarized in Table~\ref{table:dipc-comparison}. 
{\dipc} takes more time to determine differential privacy for most examples. 
In some cases, {\ourtool} significantly outperforms {\dipc} in verification time.

\begin{table}[t]
        \scriptsize
    \centering
    \renewcommand{\arraystretch}{1.5}
    \setlength\tabcolsep{4pt}
    \begin{tabular}{|lcc|cc|lc|}
    \hline
        \multirow{2}{*}{Example} &  \multirow{2}{*}{$N$} & \multirow{2}{*}{$\epsilon$} & \multicolumn{2}{c|}{Time (s)} & \multirow{2}{*}{Speedup} & \multirow{2}{*}{DP?} \\ \cline{4-5}    
        & & & {\dipc~\cite{bcjsv20}} & {\ourtool} & &\\
        \hline
        \multirow{4}{*}{\LaplaceBelowThreshold}
         & 1 & 0.5 & 25 & 1 & 25.0 & $\checkmark$ \\ 
        
        & 2 & 0.5 & 106 & 32 & 3.31 & $\checkmark$ \\ 
       
        & 3 & 0.5 & 578 & 279 & 2.07 & $\checkmark$ \\ 

        & 4 & 0.5 & 2850 & 1638 & 1.74 & $\checkmark$ \\ \hline
        \multirow{2}{*}{\LaplaceNoisyMax }
        & 3 & 1 & 278 & 166 & 1.67 & $\checkmark$ \\ 
        & 3 & 0.5 & 311 & 152 & 2.05 & $\checkmark$ \\ \hline
        \multirow{2}{*}{\LaplaceNoisyMin} & 3 & 1 & 180 & 165 & 1.09 & $\checkmark$ \\ 
         & 3 & 0.5 & 286 & 154 & 1.86 & $\checkmark$ \\ \hline
         \NonPrivBelowThresholdLaplace-$4$ & 2 & 1 & 80 & 167 & 0.48 & $\times$ \\ 
        \NonPrivBelowThresholdLaplace-$5$ & 2 & 0.5 & 7 & 1 & 7.0 & $\times$ \\ 
        \NonPrivBelowThresholdLaplace-$6$ & 3 & 1 & 526 & 1075 & 0.49 & $\times$ \\  \hline 
    \end{tabular}
    \captionsetup{font=footnotesize,labelfont=footnotesize}\caption{\footnotesize Summary of comparison with {\dipc}. The table reports performance for both tools. The columns are as follows: $N$ denotes the input size of the program. Time indicates the average time (in seconds) to verify differential privacy over three runs. DP? indicates whether the program is differentially private. Speedup represents the ratio of the time taken by {\dipc} to that of {\ourtool}, indicating the relative performance gain. Differential privacy checks were performed with $\delta = 0$. For all examples in the table, $\Depsilon=\epsilon$.}
     \label{table:dipc-comparison}
\end{table}

\subsection{Discussion}
Some salient insights from our experiments are as follows.
\begin{enumerate}[left= 0pt, topsep=1pt]
    \item {\ourtool} can determine whether programs in {\ourlang} are differentially private or not differentially private for several interesting examples. If a program is not differentially private, it provides a counterexample.
    \item {\ourtool} demonstrates scalability, handling input sizes of up to 25 for some {\SVT} variants in the case of a single input pair.
    \item The optimization algorithm significantly reduces the running time, achieving a substantial decrease, and lowers the nested depth of integral expressions.
    \item     Verification involving the Laplace distribution takes longer than the Gaussian distribution. 
    The Laplace distribution does not have a holomorphic extension to the complex numbers as it involves an absolute value term. This makes integration in {\FLINT} more computationally intensive.

\end{enumerate}

%% file: sec_related.tex
Differential privacy was first introduced in~\cite{dmns06}. For a comprehensive introduction to the topic, techniques, and results, consult the recent book~\cite{DR14} and survey~\cite{local-dp-survey}. Industry implementation of differential privacy include U.S. Census Bureau's LEHD OnTheMap tool~\cite{census}, Google's RAPPOR system
~\cite{rappor}, Apple's DP implementation~\cite{apple-white,apple-patent}, and Microsoft's Telemetry collection~\cite{microsoft}.

The subtle nature of the correctness proofs of differential privacy has prompted an interest in developing automated techniques to verify them. The main approaches to verification include the use of type systems~\cite{dfuzz,AGHK18,rp10,AmorimAGH15,ZK17,CheckDP}, probabilistic coupling~\cite{BKOZ13,bgghs16,BFGGHS16,AH18}, using shadow executions~\cite{WDWKZ19}, and simulation-based methods~\cite{TschantzKD11}, and machine-checked proofs~\cite{sato-popl,sato-ahl}. 

Automated methods for verifying privacy include the use of hypothesis testing~\cite{DingWWZK18}, symbolic differentiation~\cite{BichselGDTV18}, and program analysis~\cite{CheckDP,bcjsv20,csvb23}. Notably,~\cite{CheckDP,csvb23} even allow for verification with unbounded inputs, and for any $\epsilon>0.$ However, they do not allow sampling from Gaussians, and verify only pure differential privacy (i.e., $\delta=0$).

Probabilistic model checking-based approaches are used in~\cite{chatzGP14,LiuWZ18,ChistikovKMP20}, where it is assumed that the program is given as a finite Markov Chain, $\epsilon$ is fixed to a concrete value. Sampling from continuous random variables is not allowed in ~\cite{chatzGP14,LiuWZ18,ChistikovKMP20}.
Almost all the automated methods discussed so far are of checking $\epsilon$-differential privacy; none work for $(\epsilon,\delta)$-differential privacy, except for~\cite{bcjsv20}. 

The decision problem of checking $(\epsilon,\delta)$-differential privacy (and therefore also $\epsilon$-differential privacy) was studied in~\cite{bcjsv20} where the problem was shown to be undecidable in general, and a decidable sub-class of programs that sample from Laplacians was identified. Smaller classes of algorithms that use sampling from Laplacians and comparison of sampled values were identified in~\cite{csv21,csvb23} where it was shown that for them, the verification for \emph{unbounded} inputs is decidable. The complexity of deciding differential privacy for randomized Boolean circuits and while programs is shown to be $\mathbf{coNP^{\#P}}$-complete and $\mathbf{PSPACE}$-complete in~\cite{GaboardiNP19} and~\cite{Gaboardi22} respectively when the number of inputs is finite, probabilistic choices are fair coin tosses and $\eulerv{\epsilon}$ is a rational number. 

 There is often a trade-off between privacy and utility (or accuracy) in differential privacy algorithms. The complementary problem of automatically evaluating their accuracy claims has received attention in the literature \cite{Ba16,SmithHA19,bcksv21}.

None of the automated verification approaches discussed above verify differential privacy of programs that sample from  Gaussians. 

%% file: sec_conclusions.tex
We addressed the problem of verifying differential privacy for parameterized programs 
that support sampling from Gaussian and Laplace distributions, and operate over finite input and output domains. 
For a class of loop-free programs called {\ourlang}, we showed that the problem of determining if a given program $\Prog$, with privacy parameter $\epsilon$, is $(\Depsilon, \delta)$-differential privacy for given rational values $\epsilon > 0$, $\Depsilon > 0$, and $\delta \in \unit$ is almost decidable: it is decidable for all values of $\delta$ in $\unit$, except for a finite set of exceptional values determined by $\Prog$, $\epsilon$, $\Depsilon$, and the adjacency relation. We establish this through an algorithm, {\VerifyDP}, which outputs one of three results: $\diffprivate$, $\notdiffprivate$, or $\unresolved$. Our implementation of {\VerifyDP} leverages the {\FLINT} library for evaluating definite integrals and incorporates several performance optimizations, such as reducing the nesting depth of integrals to enhance scalability on practical benchmarks. The algorithm is implemented in our tool {\ourtool} and has been empirically evaluated on a variety of examples drawn from the literature.

For future work, it would be interesting to explore extensions in three directions: (i) allowing unbounded loops as well as non-linear functions of real variables in programs, (ii) generalizing input and output domains to real or rational values, and (iii) enabling the privacy parameter $\epsilon$ to range over an interval, with $\Depsilon$ (or $\delta$) specified as a function of $\epsilon$ (or $\Depsilon$, respectively).

%% file: app_examples.tex
We present a short description and pseudo-code of the examples from our benchmark suite.

\subsection{{\SVT} variants} We have the following variants: {\BelowThreshold}, {\LaplaceBelowThreshold}, {\MixOneBelowThreshold}, and {\MixTwoBelowThreshold}. These are similar algorithms that differ only in the distributions from which they sample noise. {\BelowThreshold} samples both the threshold and the queries from a Gaussian distribution. {\LaplaceBelowThreshold} samples both from a Laplace distribution. {\MixOneBelowThreshold} samples the threshold from a Gaussian distribution and the queries from a Laplace distribution, while {\MixTwoBelowThreshold} does the opposite. These algorithms output~$\top$ when the noisy query result is less than or equal to the noisy threshold; otherwise, they output~$\bot$. We also have non-private variants of {\BelowThreshold}: {\NonPrivBelowThresholdOne} and {\NonPrivBelowThresholdTwo}. {\NonPrivBelowThresholdOne} compares noisy queries with a non-noisy threshold, whereas {\NonPrivBelowThresholdTwo} compares a noisy threshold with non-noisy queries. Additionally, we consider four non-private variants of {\LaplaceBelowThreshold}, borrowed from~\cite{bcjsv20}: {\NonPrivBelowThresholdLaplace-3}, {\NonPrivBelowThresholdLaplace-4}, {\NonPrivBelowThresholdLaplace-5}, and {\NonPrivBelowThresholdLaplace-6}.

Another set of examples of {\SVT} variants includes {\AboveThreshold}, {\LaplaceAboveThreshold}, {\MixOneAboveThreshold}, and {\MixTwoAboveThreshold}. These algorithms are also distinguished by the distributions from which they sample noise. {\AboveThreshold} samples both the threshold and the queries from a Gaussian distribution. {\LaplaceAboveThreshold} samples both from a Laplace distribution. {\MixOneAboveThreshold} samples the threshold from a Gaussian distribution and the queries from a Laplace distribution, while {\MixTwoAboveThreshold} does the opposite. These algorithms output~$\top$ when the noisy query result is greater than or equal to the noisy threshold; otherwise, they output~$\bot$. We also have non-private versions of {\AboveThreshold}: {\NonPrivAboveThresholdOne} and {\NonPrivAboveThresholdTwo}. {\NonPrivAboveThresholdOne} compares noisy queries with a non-noisy threshold, whereas {\NonPrivAboveThresholdTwo} compares a noisy threshold with non-noisy queries.

\begin{algorithm}
\DontPrintSemicolon
\SetAlgoLined

\KwIn{$q[1:N]$}
\KwOut{$out[1:N]$}
\;
$\rv_T \gets \gauss{T,  \frac {2 \Delta}{\epsilon}}$\;

  \For{$i\gets 1$ \KwTo $N$}
  {
    $\rv \gets \gauss{q[i], \frac{4\Delta}{\epsilon}}$\;
    \uIf{$ \rv \geq \rv_T $}{
      $out[i] \gets \top$\;
      exit\;
    }\Else{
      $out[i] \gets \bot$}
  }

\caption{\BelowThreshold}
\label{algo:BelowThreshold}
\end{algorithm}

\begin{algorithm}
\DontPrintSemicolon
\SetAlgoLined

\KwIn{$q[1:N]$}
\KwOut{$out[1:N]$}
\;
$\rv_T \gets \lap{T,  \frac {2 \Delta}{\epsilon}}$\;

  \For{$i\gets 1$ \KwTo $N$}
  {
    $\rv \gets \lap{q[i], \frac{4\Delta}{\epsilon}}$\;
    \uIf{$ \rv \geq \rv_T $}{
      $out[i] \gets \top$\;
      exit\;
    }\Else{
      $out[i] \gets \bot$}
  }

\caption{\LaplaceBelowThreshold}
\label{algo:LaplaceBelowThreshold}
\end{algorithm}

\begin{algorithm}
\DontPrintSemicolon
\SetAlgoLined

\KwIn{$q[1:N]$}
\KwOut{$out[1:N]$}
\;
$\rv_T \gets \lap{T,  \frac {2 \Delta}{\epsilon}}$\;
 $count \gets 0$\;
  \For{$i\gets 1$ \KwTo $N$}
  {
    $\rv \gets \gauss{q[i], \frac{4\Delta}{\epsilon}}$\;

    \uIf{$ \rv \geq \rv_T $}{
      $out[i] \gets \top$\;
      $count \gets count + 1$\;
       \If{$count\geq c$} {exit} 
      }
    \Else{
      $out[i] \gets \bot$}
  }

\caption{\MixOneBelowThreshold}
\label{fig:svtmixone}
\end{algorithm}

\begin{algorithm}
\DontPrintSemicolon
\SetAlgoLined

\KwIn{$q[1:N]$}
\KwOut{$out[1:N]$}
\;
$\rv_T \gets \gauss{T,  \frac {2 \Delta}{\epsilon}}$\;
 $count \gets 0$\;
  \For{$i\gets 1$ \KwTo $N$}
  {
    $\rv \gets \lap{q[i], \frac{4\Delta}{\epsilon}}$\;

    \uIf{$ \rv \geq \rv_T $}{
      $out[i] \gets \top$\;
      $count \gets count + 1$\;
       \If{$count\geq c$} {exit} 
      }
    \Else{
      $out[i] \gets \bot$}
  }

\caption{\MixTwoBelowThreshold}
\label{fig:svtmixtwo}
\end{algorithm}

\begin{algorithm}
\DontPrintSemicolon
\SetAlgoLined
\KwIn{$q[1:N]$}
\KwOut{$out[1:N]$}
\;
$\rv_T \gets T$\;

  \For{$i\gets 1$ \KwTo $N$}
  {
    $\rv \gets \gauss{q[i],  \frac {2 \Delta}{\epsilon}}$\;
    \uIf{$ \rv \geq \rv_T $}{
      $out[i] \gets \top$\;
      exit\;
    }\Else{
      $out[i] \gets \bot$}
  }
\caption{\NonPrivBelowThresholdOne}
\label{alg:NonPrivBelowThresholdOne}
\end{algorithm}

\begin{algorithm}
\DontPrintSemicolon
\SetAlgoLined
\KwIn{$q[1:N]$}
\KwOut{$out[1:N]$}
\;
$\rv_T \gets \gauss{T,  \frac {2 \Delta}{\epsilon}}$\;

  \For{$i\gets 1$ \KwTo $N$}
  {
    $\rv \gets q[i]$\;
    \uIf{$ \rv \geq \rv_T $}{
      $out[i] \gets \top$\;
      exit\;
    }\Else{
      $out[i] \gets \bot$}
  }
\caption{\NonPrivBelowThresholdTwo}
\label{alg:NonPrivBelowThresholdTwo}
\end{algorithm}


  \begin{algorithm}
  \DontPrintSemicolon
  \KwIn{$q[1:N]$}
  \KwOut{$out[1:N]$}
  $\rv_T \gets \lap{ T, \frac  {4 \Delta}{\epsilon} }$\;
  $count \gets 0$\;
  \For{$i\gets 1$ \KwTo $N$}
  {
    $\rv\gets \lap{q[i], \frac  {4\Delta}{3\epsilon}}$\;
    $\bv\gets \rv \geq \rv_T$\;
    \uIf{$b$}{
      $out[i] \gets \top$\;
      $count \gets count + 1$\;
      \If{$count\geq c$} {\exit} 
      }
    \Else{
      $out[i] \gets \bot$}
  }
    \caption{\NonPrivBelowThresholdLaplace-$4$}\label{algo:svt4}
  \end{algorithm}

  \begin{algorithm}
  \DontPrintSemicolon
  \KwIn{$q[1:N]$}
  \KwOut{$out[1:N]$}

  $\rv_T \gets \lap{T, \frac  {2 \Delta}{\epsilon}}$\;
  \For{$i\gets 1$ \KwTo $N$}
  {
    $\rv\gets q[i]$\;
    $\bv\gets \rv \geq \rv_T$\;
    \uIf{$b$}{
      $out[i] \gets \top$\;
      }
    \Else{
      $out[i] \gets \bot$}
  }
   \caption{\NonPrivBelowThresholdLaplace-$5$}\label{algo:svt5}
  \end{algorithm}

  \begin{algorithm}
  \DontPrintSemicolon
  \KwIn{$q[1:N]$}
  \KwOut{$out[1:N]$}

  $\rv_T \gets \lap{  T,\frac  {2 \Delta}{\epsilon} }$\;
  \For{$i\gets 1$ \KwTo $N$}
  {
    $\rv\gets \lap{q[i], \frac{2\Delta}{\epsilon}}$\;
    $\bv\gets \rv \geq \rv_T$\;
    \uIf{$b$}{
      $out[i] \gets \top$\;
      }
    \Else{
      $out[i] \gets \bot$}
  }
   \caption{\NonPrivBelowThresholdLaplace-$6$}\label{algo:svt6}
  \end{algorithm}

\begin{algorithm}
\DontPrintSemicolon
\SetAlgoLined
\KwIn{$q[1:N]$}
\KwOut{$out[1:N]$}
\;
$\rv_T \gets \gauss{T,  \frac {2 \Delta}{\epsilon}}$\;

  \For{$i\gets 1$ \KwTo $N$}
  {
    $\rv \gets \gauss{q[i], \frac{4\Delta}{\epsilon}}$\;
    \uIf{$ \rv \leq \rv_T $}{
      $out[i] \gets \top$\;
      exit\;
    }\Else{
      $out[i] \gets \bot$}
  }

\caption{\AboveThreshold}
\label{algo:AboveThreshold}
\end{algorithm}

\begin{algorithm}
\DontPrintSemicolon
\SetAlgoLined

\KwIn{$q[1:N]$}
\KwOut{$out[1:N]$}
\;
$\rv_T \gets \lap{T,  \frac {2 \Delta}{\epsilon}}$\;

  \For{$i\gets 1$ \KwTo $N$}
  {
    $\rv \gets \lap{q[i], \frac{4\Delta}{\epsilon}}$\;
    \uIf{$ \rv \leq \rv_T $}{
      $out[i] \gets \top$\;
      exit\;
    }\Else{
      $out[i] \gets \bot$}
  }

\caption{\LaplaceAboveThreshold}
\label{algo:LaplaceAboveThreshold}
\end{algorithm}

\begin{algorithm}
\DontPrintSemicolon
\SetAlgoLined

\KwIn{$q[1:N]$}
\KwOut{$out[1:N]$}
\;
$\rv_T \gets \lap{T,  \frac {2 \Delta}{\epsilon}}$\;
 $count \gets 0$\;
  \For{$i\gets 1$ \KwTo $N$}
  {
    $\rv \gets \gauss{q[i], \frac{4\Delta}{\epsilon}}$\;

    \uIf{$ \rv \leq \rv_T $}{
      $out[i] \gets \top$\;
      $count \gets count + 1$\;
       \If{$count\geq c$} {exit} 
      }
    \Else{
      $out[i] \gets \bot$}
  }

\caption{\MixOneAboveThreshold}
\label{fig:MixOneAboveThreshold}
\end{algorithm}

\begin{algorithm}
\DontPrintSemicolon
\SetAlgoLined

\KwIn{$q[1:N]$}
\KwOut{$out[1:N]$}
\;
$\rv_T \gets \gauss{T,  \frac {2 \Delta}{\epsilon}}$\;
 $count \gets 0$\;
  \For{$i\gets 1$ \KwTo $N$}
  {
    $\rv \gets \lap{q[i], \frac{4\Delta}{\epsilon}}$\;

    \uIf{$ \rv \leq \rv_T $}{
      $out[i] \gets \top$\;
      $count \gets count + 1$\;
       \If{$count\geq c$} {exit} 
      }
    \Else{
      $out[i] \gets \bot$}
  }

\caption{\MixTwoAboveThreshold}
\label{fig:MixTwoAboveThreshold}
\end{algorithm}

\begin{algorithm}
\DontPrintSemicolon
\SetAlgoLined
\KwIn{$q[1:N]$}
\KwOut{$out[1:N]$}
\;
$\rv_T \gets T$\;

  \For{$i\gets 1$ \KwTo $N$}
  {
    $\rv \gets \gauss{q[i],  \frac {2 \Delta}{\epsilon}}$\;
    \uIf{$ \rv \leq \rv_T $}{
      $out[i] \gets \top$\;
      exit\;
    }\Else{
      $out[i] \gets \bot$}
  }
\caption{\NonPrivAboveThresholdOne}
\label{alg:NonPrivAboveThresholdOne}
\end{algorithm}

\begin{algorithm}
\DontPrintSemicolon
\SetAlgoLined

\KwIn{$q[1:N]$}
\KwOut{$out[1:N]$}
\;
$\rv_T \gets \gauss{T,  \frac {2 \Delta}{\epsilon}}$\;

  \For{$i\gets 1$ \KwTo $N$}
  {
    $\rv \gets q[i]$\;
    \uIf{$ \rv \leq \rv_T $}{
      $out[i] \gets \top$\;
      exit\;
    }\Else{
      $out[i] \gets \bot$}
  }

\caption{\NonPrivAboveThresholdTwo}
\label{algo:NonPrivAboveThreshold}
\end{algorithm}

\subsection{{\NoisyMin} and {\NoisyMax}}
\label{app:nm-code}

We have four variants of {\NoisyMin} and {\NoisyMax}: {\NoisyMinGauss}, {\LaplaceNoisyMin}, {\NoisyMaxGauss}, and {\LaplaceNoisyMax}. {\NoisyMinGauss} and {\LaplaceNoisyMin} are similar algorithms that differ only in the noise distribution: {\NoisyMinGauss} uses Gaussian noise, whereas {\LaplaceNoisyMin} uses Laplace noise. These algorithms add noise to each query and perform an \emph{argmin} operation, returning the index of the noisy minimum value.

Similarly, {\NoisyMaxGauss} and {\LaplaceNoisyMax} are also similar algorithms that differ only in the noise distribution: {\NoisyMaxGauss} uses Gaussian noise, whereas {\LaplaceNoisyMax} uses Laplace noise. These algorithms add noise to each query and perform an \emph{argmax} operation, returning the index of the noisy maximum value.

\begin{algorithm}
\DontPrintSemicolon
\KwIn{$q[1:N]$}
\KwOut{$out$}
\;
NoisyVector $\gets []$\;
\For{$i\gets 1$ \KwTo $N$}{
NoisyVector[i] $\gets$ $\gauss{q[i], \frac{4\Delta}{\epsilon}}$
}
out $\gets$ argmax(NoisyVector)\;
\caption{\NoisyMaxGauss}\label{algo:nmax}
\end{algorithm}

\begin{algorithm}
\DontPrintSemicolon
\KwIn{$q[1:N]$}
\KwOut{$out$}
\;
NoisyVector $\gets []$\;
\For{$i\gets 1$ \KwTo $N$}{
NoisyVector[i] $\gets$ $\gauss{q[i], \frac{4\Delta}{\epsilon}}$
}
out $\gets$ argmin(NoisyVector)\;
\caption{\NoisyMinGauss}\label{algo:noisy-min}
\end{algorithm}

  \begin{algorithm}
  \DontPrintSemicolon
  \KwIn{$q[1:N]$}
  \KwOut{$out$}
  \;
  NoisyVector $\gets []$\;
  \For{$i\gets 1$ \KwTo $N$}{
    NoisyVector[i] $\gets$ $\lap{q[i], \frac{2}{\epsilon}}$
  }
  out $\gets$ argmax(NoisyVector)\;

    \caption{\LaplaceNoisyMax}\label{algo-nmax7}
  \end{algorithm}

  \begin{algorithm}
  \DontPrintSemicolon
  \KwIn{$q[1:N]$}
  \KwOut{$out$}
  \;
  NoisyVector $\gets []$\;
  \For{$i\gets 1$ \KwTo $N$}{
    NoisyVector[i] $\gets$ $\lap{q[i], \frac{2}{\epsilon}}$
  }
  out $\gets$ argmin(NoisyVector)\;

    \caption{\LaplaceNoisyMin}\label{algo:nmin-laplace}
  \end{algorithm}

\subsection{$k$-{\minmax} and $m$-{\Range}} 
The $k$-{\minmax} algorithm (for $k \geq 2$) perturbs the first $k$ queries with Laplace noise, computes the noisy maximum and minimum, and then checks whether each subsequent noisy query falls within this range; if not, the algorithm exits. The $m$-{\Range} algorithm perturbs $2m$ thresholds that define a rectangle of $m$ dimensions and checks whether noisy queries lie within these noisy limits.

\begin{algorithm}
\DontPrintSemicolon
\KwIn{$q[1:m]$}
\KwOut{$out[1:Nm]$}
\;

\For{$j\gets 1$ \KwTo $m$}{
    $\s{low[j]} \gets \lap{T_1[j], \frac{4m}{\epsilon}}$\;
    $\s{high[j]} \gets \lap{T_2[j], \frac{4m}{\epsilon}}$\;
    $out[j] \gets \s{cont}$\;
}

\For{$i\gets 1$ \KwTo $N$}
{
    \For{$j\gets 1$ \KwTo $m$}{
        $\rv\gets \lap{q[m(i-1) + j],\frac{4}{\epsilon}}$\;
        \uIf{$(\rv \geq \s{low[j]}) \wedge (\rv < \s{high[j]})$}{
        $out[m(i-1) + j] \gets \s{cont}$\;  }
        \ElseIf{$((\rv \geq \s{low[j]}) \wedge (\rv > \s{high[j]}))$}{
          $out[m(i-1) + j] \gets \top$\;
          exit
        } 
        \ElseIf{$((\rv < \s{low[j]}) \wedge (\rv < \s{high[j]}))$}{
          $out[m(i-1) + j] \gets \bot $\;
          exit
        }
    }
}
\caption{\mRange}
\label{algo:mRange}
\end{algorithm}

\begin{algorithm}
\DontPrintSemicolon

\KwIn{$q[1:N]$}
\KwOut{$out[1:N]$}
\;

$\s{min}, \s{max} \gets \gauss{q[1], \frac{4k}{\epsilon})}$\;

\For{$i\gets 2$ \KwTo $k$}
{
    
    $\rv\gets \gauss{q[i], \frac{4k}{\epsilon}}$\;
    
    \uIf{$(\rv > \s{max}) \wedge (\rv > \s{min})$}{
        $\s{max} \gets \rv$\;
    }\ElseIf{$(\rv < \s{min}) \wedge (\rv < \s{max})$ } {
        $\s{min} \gets \rv$\;
        
    }
    $out[i] \gets \mathsf{read}$
}

\For{$i\gets k+1$ \KwTo $N$}
{
    $\rv\gets \gauss{q[i], \frac{4}{\epsilon}}$\;
    \uIf{$(\rv \geq \s{min}) \wedge (\rv < \s{max})$}{
      $out[i] \gets \bot$}
    \uElseIf{$(\rv \geq \s{min}) \wedge (\rv \geq \s{max})$}{
      $out[i] \gets \top$\;
      exit
    } \ElseIf{$(\rv < \s{min}) \wedge (\rv < \s{max})$}{
      $out[i] \gets \bot$\;
      exit
    }

}
\caption{\kminmax}
\label{fig:k-min-max}
\end{algorithm}

%% file: app_experiments.tex
Here, we present the complete experimental results. Table~\ref{table:full-examples} shows the performance results of our benchmark suite. Table~\ref{table:full-optimization} illustrates the impact of optimization on {\BelowThreshold}, {\LaplaceBelowThreshold}, {\MixOneBelowThreshold}, and {\NoisyMaxGauss}. Table~\ref{table:full-dipc-comparison} provides a comparison with the {\dipc} tool. Table~\ref{table:impact-of-delta-eps} demonstrates the effect of varying $\epsilon$ and $\delta$. Section~\ref{sec:comparison-checkdp} discusses the comparison with {\CheckDP}~\cite{CheckDP}.

\begin{table}[htp]
\centering
 \scriptsize
 \renewcommand{\arraystretch}{1.5}
    \setlength\tabcolsep{4pt}
\begin{tabular}{c|cccccccc}
\toprule
\diagbox{$\epsilon$}{$\delta$} & 0.01 & 0.05 & 0.10 & 0.20 & 0.30 & 0.40 & 0.50 & 0.60 \\
\midrule
0.05 & 7.13 & 6.37 & 6.59 & 6.24 & 6.38 & 8.35 & 7.95 & 6.68 \\
0.08 & 6.71 & 7.62 & 8.23 & 6.82 & 7.35 & 7.59 & 7.49 & 7.58 \\
0.09 & 7.35 & 7.05 & 6.51 & 6.45 & 7.00 & 9.25 & 9.76 & 9.49 \\
0.10 & 6.80 & 6.40 & 6.45 & 6.50 & 6.24 & 6.09 & 6.12 & 6.01 \\
0.20 & 6.85 & 6.17 & 5.88 & 5.96 & 6.05 & 6.35 & 5.96 & 5.83 \\
0.30 & 5.89 & 5.84 & 5.81 & 6.27 & 6.05 & 5.82 & 5.93 & 5.78 \\
0.40 & 5.80 & 6.35 & 6.11 & 5.99 & 6.01 & 6.08 & 5.99 & 5.96 \\
0.50 & 5.62 & 5.57 & 5.68 & 5.58 & 5.55 & 5.76 & 5.66 & 5.54 \\
0.60 & 6.06 & 7.01 & 6.53 & 6.07 & 6.06 & 5.98 & 6.14 & 6.05 \\
0.70 & 8.57 & 6.95 & 6.72 & 6.38 & 5.99 & 6.47 & 6.90 & 6.07 \\
0.80 & 6.04 & 5.98 & 5.90 & 6.01 & 6.06 & 6.00 & 5.93 & 5.97 \\
0.90 & 5.89 & 5.97 & 5.96 & 5.93 & 5.95 & 6.02 & 5.98 & 6.14 \\
1.00 & 5.64 & 5.68 & 5.70 & 5.65 & 5.67 & 5.66 & 5.63 & 5.72 \\
\bottomrule
\end{tabular}
\captionsetup{font=footnotesize,labelfont=footnotesize}\caption{\footnotesize Summary of the impact of varying $\epsilon$ and $\delta$ on the {\AboveThreshold} example with an input size of $N = 5$. In all cases, we used $\Depsilon = \epsilon$.}
\label{table:impact-of-delta-eps}
\end{table}

\begin{table}[htp]
        \scriptsize
    \centering
    \renewcommand{\arraystretch}{1.5}
    \setlength\tabcolsep{4pt}
    \begin{tabular}{|lcc|cc|lc|}
    \hline
        \multirow{2}{*}{Example} &  \multirow{2}{*}{$N$} & \multirow{2}{*}{$\epsilon$} & \multicolumn{2}{c|}{Time} & \multirow{2}{*}{Speedup} & \multirow{2}{*}{DP?} \\ \cline{4-5}    
        & & & {\dipc~\cite{bcjsv20}} & {\ourtool} & &\\
        \hline
        \multirow{8}{*}{\LaplaceBelowThreshold}
         & 1 & 1 & 52 & 1 & 52.0 & $\checkmark$ \\ 
         & 1 & 0.5 & 25 & 1 & 25.0 & $\checkmark$ \\ 
        & 2 & 1 & 104 & 26 & 4.0 & $\checkmark$ \\ 
        & 2 & 0.5 & 106 & 32 & 3.31 & $\checkmark$ \\ 
        & 3 & 1 & 558 & 250 & 2.23 & $\checkmark$ \\ 
        & 3 & 0.5 & 578 & 279 & 2.07 & $\checkmark$ \\ 
        & 4 & 1 & 2814 & 1481 & 1.9 & $\checkmark$ \\ 
        & 4 & 0.5 & 2850 & 1638 & 1.74 & $\checkmark$ \\ \hline
         \multirow{8}{*}{\LaplaceAboveThreshold}
         & 1 & 1 & 29 & 1 & 29.0 & $\checkmark$ \\ 
        & 1 & 0.5 & 23 & 1 & 23.0 & $\checkmark$ \\ 
        & 2 & 1 & 145 & 25 & 5.8 & $\checkmark$ \\ 
        & 2 & 0.5 & 163 & 22 & 7.41 & $\checkmark$ \\ 
        & 3 & 1 & 906 & 227 & 3.99 & $\checkmark$ \\ 
        & 3 & 0.5 & 1134 & 204 & 5.56 & $\checkmark$ \\ 
        & 4 & 1 & 4317 & 1684 & 2.56 & $\checkmark$ \\ 
        & 4 & 0.5 & 4887 & 1285 & 3.8 & $\checkmark$ \\ \hline
        \multirow{2}{*}{\LaplaceNoisyMax }
        & 3 & 1 & 278 & 166 & 1.67 & $\checkmark$ \\ 
        & 3 & 0.5 & 311 & 152 & 2.05 & $\checkmark$ \\ \hline
        \multirow{2}{*}{\LaplaceNoisyMin} & 3 & 1 & 180 & 165 & 1.09 & $\checkmark$ \\ 
         & 3 & 0.5 & 286 & 154 & 1.86 & $\checkmark$ \\ \hline
         \NonPrivBelowThresholdLaplace-$4$ & 2 & 1 & 80 & 167 & 0.48 & $\times$ \\ 
        \NonPrivBelowThresholdLaplace-$5$ & 2 & 0.5 & 7 & 1 & 7.0 & $\times$ \\ 
        \NonPrivBelowThresholdLaplace-$6$ & 3 & 1 & 526 & 1075 & 0.49 & $\times$ \\  \hline 
    \end{tabular}
    \captionsetup{font=footnotesize,labelfont=footnotesize}\caption{\footnotesize Summary of comparison with {\dipc}. The table reports performance for both tools. The columns are as follows: $N$ denotes the input size of the program. Time indicates the average time (in seconds) to verify differential privacy over three runs. DP? indicates whether the program is differentially private. Speedup represents the ratio of the time taken by {\dipc} to that of {\ourtool}, indicating the relative performance gain. Differential privacy checks were performed with $\delta = 0$. For all examples in the table, $\Depsilon=\epsilon$.}
     \label{table:full-dipc-comparison}
\end{table}

\begin{table}[htp]
    \centering
    \scriptsize
\setlength\tabcolsep{4pt}
\renewcommand{\arraystretch}{1.5}
    \begin{tabular}{|lc|ccc|ccc|}
    \hline
        \multirow{3}{*}{Example} & \multirow{3}{*}{$N$}  & \multicolumn{3}{c|}{Unoptimized} & \multicolumn{3}{c|}{Optimized} \\
        \cline{3-8}
        &  & Max & Avg. & \multirow{2}{*}{Time} & Max & Avg. & \multirow{2}{*}{Time}\\

        & & Depth & Depth & & Depth & Depth & \\
        \hline
        \multirow{5}{*}{\BelowThreshold}
        & 1 & 2 & 2.0 & 1.03 & 2 & 2.0 & 1.0 \\ 
        & 2 & 3 & 2.67 & 2.27 & 3 & 2.3 & 1.6 \\ 
        & 3 & 4 & 3.25 & $\timeout$ & 3 & 2.5 & 2.77 \\ 
        & 4 & 5 & 3.8 & $\timeout$ & 3 & 2.6 & 3.48 \\ 
        & 5 & 6 & 4.33 & $\timeout$ & 3 & 2.7 & 7.9 \\ \hline
        \multirow{5}{*}{\LaplaceBelowThreshold}
         & 1 & 2 & 2.0 & 1.04 & 2 & 2.0 & 1.0 \\ 
         & 2 & 3 & 2.67 & 14.21 & 3 & 2.3 & 5.1 \\ 
         & 3 & 4 & 3.25 & $\timeout$ & 3 & 2.5 & 12.79 \\ 
         & 4 & 5 & 3.8 & $\timeout$ & 3 & 2.6 & 26.2 \\ 
        & 5 & 6 & 4.33 & $\timeout$ & 3 & 2.7 & 47.6 \\ \hline
        \multirow{5}{*}{\MixOneBelowThreshold}
         & 1 & 2 & 2.0 & 1.01 & 2 & 2.0 & 0.94 \\ 
        & 2 & 3 & 2.67 & 7.62 & 3 & 2.3 & 2.4 \\ 
         & 3 & 4 & 3.25 & $\timeout$ & 3 & 2.5 & 7.88 \\ 
        & 4 & 5 & 3.8 & $\timeout$ & 3 & 2.6 & 11.8 \\ 
        & 5 & 6 & 4.33 & $\timeout$ & 3 & 2.7 & 19.5 \\ \hline
        \multirow{4}{*}{\NoisyMaxGauss}
         & 2 & 2 & 2.0 & 0.98 & 2 & 2.0 & 1.0 \\ 
         & 3 & 3 & 3.0 & 3.53 & 3 & 2.5 & 1.6 \\ 
        & 4 & 4 & 4.0 & $\timeout$ & 4 & 3.0 & 37.8 \\ 
        & 5 & 5 & 5.0 & $\timeout$ & 5 & 3.5 & $\timeout$ \\ 
        \hline
    \end{tabular}
    \captionsetup{font=footnotesize,labelfont=footnotesize}\caption{\footnotesize Summary of optimization results for {\ourtool}. The columns in the table are as follows: $N$ represents the input size for the program. Time refers to the time taken to check differential privacy for a single pair, measured in seconds and averaged over three executions. $\timeout$ indicates a timeout (exceeding 10 minutes). Avg. Depth refers to the average nested depth of integrals across all executions. Max Depth refers to the maximum nested depth among all executions. The optimized algorithm corresponds to Algorithm~\ref{alg:optimization}. Differential privacy checks were performed with $\epsilon = 0.5$ and $\delta = 0.01$.}
    \label{table:full-optimization}
\end{table}

\begin{table}[htp]
    \scriptsize
    \centering
    \renewcommand{\arraystretch}{1.5}
    \setlength\tabcolsep{4pt}
    \begin{tabular}{|lc|ccc|cc|cc|}
    \hline
        \multirow{2}{*}{Example} & \multirow{2}{*}{$N$} & Final & \multirow{2}{*}{$\overline{|G|}$} & Avg. &
        \multicolumn{2}{c|}{Single Pair} & \multicolumn{2}{c|}{All Pairs} \\\cline{6-9}
         
        & & States &  & Depth & DP? & Time & DP? & Time \\
        \hline
         \multirow{3}{*}{\BelowThreshold}  & 2 & 3 & 1.7 & 2.3 & $\checkmark$ & 1.6 & $\checkmark$ & 2.4 \\ 
         & 5 & 6 & 3.3 & 2.7 & $\checkmark$ & 7.9 & $\checkmark$ & 76.7 \\ 
         & 25 & 26 & 13.5 & 2.9 & $\checkmark$ & 441.3 & $\na$ & $\experimenterror$ \\ \hline

         \multirow{2}{*}{\NonPrivBelowThresholdOne}
         & 5 & 6 & 3.3 & 1.7 & $\times$ & 1.0 & $\times$ & 1.4 \\ 
         & 6 & 7 & 3.9 & 1.7 & $\times$ & 1.0 & $\times$ & 2.1 \\ \hline
         \multirow{2}{*}{\NonPrivBelowThresholdTwo}
         & 3 & 4 & 2.2 & 1.5 & $\times$ & 1.0 & $\times$ & 1.0 \\ 
         & 6 & 7 & 3.9 & 1.7 & $\times$ & 1.0 & $\times$ & 1.1 \\ \hline
         \multirow{3}{*}{\AboveThreshold}
         & 2 & 3 & 1.7 & 2.3 & $\checkmark$ & 1.3 & $\checkmark$ & 1.5 \\ 
         & 5 & 6 & 3.3 & 2.7 & $\checkmark$ & 5.7 & $\checkmark$ & 72.3 \\ 
         & 25 & 26 & 13.5 & 2.9 & $\checkmark$ & 501.6 & $\na$ & $\experimenterror$ \\ \hline
         \multirow{2}{*}{\NonPrivAboveThresholdOne}
         & 5 & 6 & 3.3 & 1.7 & $\times$ & 1.0 & $\times$ & 1.2 \\ 
         & 6 & 7 & 3.9 & 1.7 & $\times$ & 1.0 & $\times$ & 1.7 \\ \hline
         \multirow{2}{*}{\NonPrivAboveThresholdTwo}
         & 3 & 4 & 2.2 & 1.5 & $\times$ & 1.0 & $\times$ & 0.9 \\ 
         & 6 & 7 & 3.9 & 1.7 & $\times$ & 1.0 & $\times$ & 1.0 \\ \hline
         \multirow{3}{*}{\LaplaceBelowThreshold} & 2 & 3 & 1.7 & 2.3 & $\checkmark$ & 5.1 & $\checkmark$ & 8.8 \\ 
         & 5 & 6 & 3.3 & 2.7 & $\checkmark$ & 47.6 & $\na$ & $\timeout$ \\ 
         & 11 & 12 & 6.4 & 2.8 & $\checkmark$ & 500.2 & $\na$ & $\timeout$ \\ \hline
         \multirow{3}{*}{\LaplaceAboveThreshold}
         & 2 & 3 & 1.7 & 2.3 & $\checkmark$ & 4.1 & $\checkmark$ & 8.8 \\ 
         & 5 & 6 & 3.3 & 2.7 & $\checkmark$ & 46.0 & $\na$ & $\timeout$ \\ 
         & 11 & 12 & 6.4 & 2.8 & $\checkmark$ & 497.8 & $\na$ & $\timeout$ \\ \hline
         \multirow{3}{*}{\MixOneBelowThreshold}
         & 2 & 3 & 1.7 & 2.3 & $\checkmark$ & 2.4 & $\checkmark$ & 4.3 \\ 
         & 5 & 6 & 3.3 & 2.7 & $\checkmark$ & 19.5 & $\checkmark$ & 285.4 \\ 
         & 17 & 18 & 9.4 & 2.9 & $\checkmark$ & 365.7 & $\na$ & $\experimenterror$ \\ \hline
         \multirow{3}{*}{\MixOneAboveThreshold}
         & 2 & 3 & 1.7 & 2.3 & $\checkmark$ & 2.3 & $\checkmark$ & 4.1 \\ 
         & 5 & 6 & 3.3 & 2.7 & $\checkmark$ & 16.1 & $\checkmark$ & 261.0 \\ 
        & 17 & 18 & 9.4 & 2.9 & $\checkmark$ & 343.6 & $\na$ & $\experimenterror$ \\ \hline
        \multirow{3}{*}{\MixTwoBelowThreshold} & 2 & 3 & 1.7 & 2.3 & $\checkmark$ & 7.2 & $\checkmark$ & 14.4 \\ 
         & 5 & 6 & 3.3 & 2.7 & $\checkmark$ & 72.2 & $\na$ & $\timeout$ \\ 
        & 10 & 11 & 5.9 & 2.8 & $\checkmark$ & 524.6 & $\na$ & $\experimenterror$ \\ \hline
        \multirow{3}{*}{\MixTwoAboveThreshold}
         & 2 & 3 & 1.7 & 2.3 & $\checkmark$ & 7.1 & $\checkmark$ & 13.4 \\ 
         & 5 & 6 & 3.3 & 2.7 & $\checkmark$ & 67.7 & $\na$ & $\timeout$ \\ 
         & 10 & 11 & 5.9 & 2.8 & $\checkmark$ & 506.5 & $\na$ & $\experimenterror$ \\ \hline
         \multirow{3}{*}{\NoisyMaxGauss}
         & 2 & 2 & 1.0 & 2.0 & $\checkmark$ & 1.0 & $\checkmark$ & 1.3 \\ 
        & 3 & 4 & 2.0 & 2.5 & $\checkmark$ & 1.6 & $\checkmark$ & 3.1 \\ 
        & 4 & 8 & 3.0 & 3.0 & $\checkmark$ & 37.8 & $\checkmark$ & 303.0 \\ \hline 
        \multirow{3}{*}{\NoisyMinGauss} & 2 & 2 & 1.0 & 2.0 & $\checkmark$ & 1.0 & $\checkmark$ & 1.3 \\ 
        & 3 & 4 & 2.0 & 2.5 & $\checkmark$ & 1.6 & $\checkmark$ & 3.1 \\ 
        & 4 & 8 & 3.0 & 3.0 & $\checkmark$ & 36.8 & $\checkmark$ & 303.6 \\ \hline
        \multirow{2}{*}{\LaplaceNoisyMax}
         & 3 & 4 & 2.0 & 2.5 & $\checkmark$ & 13.1 & $\checkmark$ & 47.2 \\ 
        & 4 & 8 & 3.0 & 3.0 & $\na$ & $\timeout$ & $\na$ & $\timeout$ \\ \hline
        \multirow{2}{*}{\LaplaceNoisyMin}
         & 3 & 4 & 2.0 & 2.5 & $\checkmark$ & 9.8 & $\checkmark$ & 45.6 \\ 
         & 4 & 8 & 3.0 & 3.0 & $\na$ & $\timeout$ & $\na$ & $\timeout$ \\ \hline
         \multirow{3}{*}{\mRange}
         & 1 & 7 & 3.0 & 2.5 & $\checkmark$ & 1.5 & $\checkmark$ & 1.8 \\ 
        & 2 & 13 & 4.2 & 3.2 & $\checkmark$ & 171.2 & $\checkmark$ & 344.4 \\ 
        & 3 & 19 & 5.2 & 3.8 & $\na$ & $\timeout$ & $\na$ & $\timeout$ \\ \hline
        \multirow{2}{*}{\kminmax}
         & 3 & 16 & 4.0 & 3.0 & $\checkmark$ & 2.1 & $\checkmark$ & 5.7 \\ 
         & 4 & 28 & 5.1 & 3.4 & $\checkmark$ & 41.2 & $\checkmark$ & 335.7 \\ 
        \hline
    \end{tabular} \captionsetup{font=footnotesize,labelfont=footnotesize}    \caption{\footnotesize Summary of Experimental Results for \ourtool. The columns in the table are defined as follows: $N$ is the input size of the program. DP? indicates whether the program is differentially private. \textbf{Final States} denotes the number of final states. $\overline{|G|}$ and Avg. Depth, respectively, denote the average number of conditions and the average nesting depth of integral expressions, per final state. Time is the average time (in seconds) to verify differential privacy, measured over three runs. $\timeout$ indicates a timeout (exceeding 10 minutes), and $\experimenterror$ denotes a run out of memory. Differential privacy checks were performed with $\epsilon = 0.5$ and $\delta = 0.01$, except for {\NonPrivBelowThresholdOne}, which uses $\epsilon = 8$. We used $\Depsilon = 0.5$, except for {\BelowThreshold}, {\AboveThreshold}  {\MixOneBelowThreshold}, {\MixOneAboveThreshold}, {\MixTwoBelowThreshold} and {\MixTwoAboveThreshold}, where $\Depsilon = 1.24$.}
    \label{table:full-examples}
\end{table}

\subsection{Comparison with \CheckDP}
\label{sec:comparison-checkdp}
We have compared our tool with {\CheckDP}~\cite{CheckDP}. However, our tool verifies privacy only for fixed values of $\epsilon$, whereas {\CheckDP} verifies for all values of $\epsilon > 0$. 
Additionally, our tool supports checking of $(\epsilon, \delta)$-differential privacy and can handle programs with Gaussian distributions, whereas {\CheckDP} only supports $\epsilon$-differential privacy and programs with Laplace distributions. Table~\ref{table:comparison-checkdp} presents the results of the comparison.

\begin{table}[htp]
        \scriptsize
    \centering
    \renewcommand{\arraystretch}{1.5}
    \setlength\tabcolsep{4pt}
    \begin{tabular}{|lc|cc|}
        \hline
        \multirow{2}{*}{Example}   &  \multirow{2}{*}{$N$}  & \multicolumn{2}{c|}{Time} \\ \cline{3-4}
        & & {\CheckDP} & {\ourtool} \\ \hline
        \LaplaceBelowThreshold & 1 & 29.9 & 1.3 \\ 
        \LaplaceMrange & 1 & $\timeout$ & 19.3 \\ 
        \LaplaceKminmax & 3 & $\timeout$ & 200.3 \\ 
        \hline
    \end{tabular}    \captionsetup{font=footnotesize,labelfont=footnotesize}\caption{\footnotesize Summary of comparison with {\CheckDP}. The table reports performance for both tools. The columns are as follows: $N$ denotes the input size of the program. Time indicates the average time (in seconds) to verify differential privacy over three runs. Differential privacy checks were performed with $ \Depsilon = \epsilon = 0.5$ and $\delta = 0$ for {\ourtool}.}
        \label{table:comparison-checkdp}

\end{table}